\documentclass[12 pt]{article}
\usepackage[utf8]{inputenc}
\newenvironment{dedication}
  {
   \thispagestyle{empty}
   \vspace*{\stretch{1.4}}
   \itshape             
   \raggedleft          
  }
  {\par 
   \vspace{\stretch{3}} 
   \clearpage           
  }
\usepackage[a4paper, top=2.5cm, bottom=2.5cm, left=2.2cm, right=2.2cm]
{geometry}
\usepackage{amsmath,amssymb,amsfonts,bm,amscd,mathtools}
\usepackage{graphicx}
\usepackage{caption, subcaption}
\usepackage{comment}
\usepackage{cite}
\usepackage{float}
\usepackage{appendix}
\usepackage{physics}
\usepackage{amsthm}
\usepackage{tcolorbox}
\usepackage{tikz}
\usepackage{color}

\DeclarePairedDelimiter\floor{\lfloor}{\rfloor}

\newtheorem{theorem}{Theorem}[section]

\newtheorem{lemma}[theorem]{Lemma}

\def\be{\begin{equation}}
\def\ee{\end{equation}}
\def\bea{\begin{align}}
\def\eea{\end{align}}
\def\xij{X_{ij}}
\def\txij{\tilde{X}_{ij}}

\begin{document}

\baselineskip 24pt

\begin{center}

{\Large \bf Towards  Positive Geometry of Multi Scalar Field Amplitudes\\
\large Accordiohedron and Effective Field Theory}


\end{center}

\vskip .6cm
\medskip

\vspace*{4.0ex}

\baselineskip=18pt

\centerline{\large \rm Mrunmay Jagadale$^{a}$ and Alok Laddha$^{b}$}

\vspace*{4.0ex}

\centerline{\large \it ~$^a$ California Institute of Technology, Pasadena, CA 91125, USA}

\centerline{\large \it ~$^b$Chennai Mathematical Institute, Siruseri, Chennai, India}
\centerline{\large \it  Siruseri, SIPCOT IT Park}


\vspace*{1.0ex}
\centerline{\small E-mail:  mjagadal@caltech.edu, aladdha@cmi.ac.in}

\vspace*{5.0ex}

\centerline{\bf Abstract} \bigskip
The geometric structure of S-matrix encapsulated by the ``Amplituhedron program" has begun to reveal itself even in non-supersymmetric quantum field theories. Starting with the seminal work of Arkani-Hamed, Bai, He and Yan \cite{abhy1711} it is now understood that for a wide class of scalar quantum field theories, tree-level amplitudes are canonical forms associated to polytopes known as accordiohedra. Similarly the higher loop scalar integrands are canonical forms associated to so called type-D cluster polytopes for cubic interactions or  recently discovered class of polytopes termed pseudo-accordiohedron for higher order scalar interactions.

\vspace*{-0.2in}
In this paper, we continue to probe the universality of these structures for a wider class of scalar quantum field theories. More in detail, we discover new realisations of the associahedron in planar kinematic space whose canonical forms generate (colour-ordered) tree-level S matrix of external massless particles with $n-4$ massless poles and one massive pole at $m^{2}$.  The resulting amplitudes are associated to  $\lambda_{1}\, \phi_{1}^{3}\, +\, \lambda_{2}\, \phi_{1}^{2}\phi_{2}$ potential  where $\phi_{1}$  and $\phi_{2}$ are massless and massive scalar fields with bi-adjoint colour indices respectively.  We also show how in the ``decoupling limit" (where $m\, \rightarrow\, \infty, \lambda_{2}\, \rightarrow\, \infty$  such that $g\, :=\, \frac{\lambda_{2}}{m}\, =\, \textrm{finite}$) these associahedra  project onto a specific class of accordiohedron which are known to be positive geometries of amplitudes generated by $\lambda\, \phi_{1}^{3}\, +\, g\, \phi_{1}^{4}$.

\begin{dedication}
Dedicated to the memory of Nila;\\ Teacher, Mentor and Friend. 
\end{dedication}

\vfill \eject

\baselineskip 18pt

\tableofcontents



\section{Introduction}
The analysis of the structure of the S-matrix  has witnessed several striking developments in past two decades. But the  scattering amplitudes associated to scalar field theories display a ``dual persona" in many of these developments. On one hand, S-matrix of bi-adjoint scalar $\phi^{3}$ theory is a prototype for many of the more sophisticated theories such as Yang-Mills theory or Gravity\footnote{In fact the bi-adjoint $\phi^{3}$ scattering amplitude also plays the role of a building block in the so-called KLT relations}, and on the other hand the  scattering amplitudes associated to higher order scalar interactions are more intricate and do not appear to share the remarkable simplicity and elegance of the amplitudes associated to cubic interactions.

BCFW recursion relations provided first hints of such intricacies as it is only the cubic scalar interactions which are BCFW constructible using a single BCFW shift \cite{nima2008,bofeng}. However the ``brutality" of generic scalar field amplitudes manifested itself most clearly in the CHY (Cachazo He, and Yuan) formulation of the S-matrix. As was shown by CHY in \cite{chy1} and later in a series of works by Bourjaily et al \cite{damgaard1, damgaard2} the integrand for the bi-adjoint $\phi^{3}$ theory is the canonical top form in the moduli space but the integrands for other scalar theories do not appear to have any obvious geometric or cohomological characterisation. (The amplituhedron program in fact motivates us to look for lower forms on CHY moduli space or top forms on binary geometries \cite{nima-binary}, \cite{song-binary} as integrands for $\phi^{p}$ theories. For some initial attempts in the first direction we refer the reader to \cite{mrunmay1911}).

The extension of the Amplituhedron program to non super-symmetric quantum field theories and the discovery of positive geometries (more specifically convex polytopes) in kinematic space \cite{abhy1711} has shed new light on some of these issues. We now understand  how the  tree-level (and colour ordered) amplitudes of a generic massless scalar field theories are directly tied to the existence of certain very specific Positive geometries (in fact Polytopes) in the kinematic space.\footnote{Even perturbative string  amplitudes have intimate relationship with positive geometries and the associated canonical forms as discovered in the seminal paper by Mizera \cite{mizera1706}.}
The ``Amplituhedron" for a $\phi^{p}$ interaction for a generic $p$ is a combinatorial polytope known as the accordiohedron which admits convex realisations in kinematic space . \cite{pppp, pinni1810, prashanth, mrunmay1911, mrunmay3, songmatter}. Each accordiohedron defines a unique canonical form in kinematic space and a (weighted) sum over all the canonical forms of a given dimension is the $n$-point amplitude of $\phi^{p}$ theory.\footnote{Equivalently, each accordiohedron is dual to a simple polytope and the canonical form associated to the accordiohedron induces a volume measure on the kinematic space such that the sum over volumes of all the dual polytopes of a given dimension equals the tree-level amplitude.}

However, even for tree-level amplitudes, a number of puzzles remain unresolved. For an $n$ particle colour-ordered amplitude in bi-adjoint scalar $\phi^{3}$ theory, there is a unique $n-3$ dimensional associahedron and hence the amplitude is nothing but the canonical top form on the associahedron \cite{abhy1711}. For a generic scalar interaction,  even for a fixed number of external particles, there is a whole set of acccordiohedra and they all need to be accounted for (the forms associated to each of these have to be added with specific coefficients) when evaluating the tree-level amplitude. Thus the scattering form for $\phi^{p\, >\, 3}$ theory is not a d-$\log$ form but a weighted sum over $k\, >\, 1$ d-$\log$ forms.

The weights  are uniquely fixed via combinatorial properties of the accordiohedron (\cite{kojima, srivastava}) but it is not apriori clear why a specific linear combination of the d $\log$-forms is ``special" in the sense that it generates unitary and local scattering amplitudes. In fact, as was observed in \cite{mrunmay3}, if we consider the accordiohedra polytopes for polynomial scalar interactions, then a class of such accordiohedra contribute with vanishing weights and it has till date remained unclear why certain accordiohedra are redundant as far as the tree-level S-matrix is concerned.

 It has been advocated by Nima Arkani-Hamed, that reason for emergence of accordiohedron polytopes (as the amplituhedron for higher order scalar interaction) should be probed via effective field theory ideas. We can think of generating polynomial scalar interactions from cubic interactions by integrating out massive fields. In this sense Accordiohedra polytopes should ``naturally" emerge from the kinematic space Associahedra in $X_{ij}\, <<\, m^{2}$ limit.\footnote{We also thank Nemani Suryanarayana for highlighting this possibility to us in 2018.} 

In this paper, we attempt to resolve few pieces of this puzzle. We take first steps towards showing that there is an ``amplituhedron" for tree-level S matrix in which the interaction Hamiltonian includes coupling of two scalars (one of which is massive and the other massless). More in detail, we show that there exists a class of polytopes  which are topologically equivalent to an associahedron but whose boundaries come in two possible colours. We refer to such a polytope as an associahedron block. We show that an associahedron block admits realisations in kinematic space whose canonical form generates a set of tree level amplitudes of $\lambda_{1}\, \phi_{1}^{3} + \lambda_{2}\, \phi_{1}^{2}\phi_{2}$ theories, where $\phi_{1},\, \phi_{2}$ are two species of scalar fields. The set consists of $n$-point amplitudes in which all the external states are massless and the amplitude is expanded upto $\lambda^{2}$.  We distinguish the two fields by taking $\phi_{1}$ to be massless and $\phi_{2}$ to be massive.  We then show that there is a way to  ``geometrize" the effect of integrating out massive modes to leading order in $\frac{1}{m}$ ($m$ being mass of $\phi_{2}$)  on positive geometries and show that this leads to the accordiohedron polytope 

\emph{Our primary result in the paper is the following} : Up to order $\lambda_{2}^{2}$ perturbative expansion of the $n$ point (colour-ordered) tree-level amplitude in which all the external particles are massless  is a sum over canonical forms of a set of  polytopes which are realisations of associahedra in the positive region of kinematic space. We derive a formula which computes the (color-ordered) amplitude for the multi-scalar theory in terms of canonical forms associated to ABHY-type realisations of associahedron blocks.

In particular, let $A_{n-3}$ be an ABHY associahedron in the positive region of planar kinematic space, ${\cal K}_{n}^{+}$. As we show below, for each $(i,j)$ such that $\{\, \vert i - j\vert$ modulo $n\, \in\, 1,\, \dots,\, \frac{n}{2}\, \}$ there exists a polytope that we call associahedron block and denote as $A_{n-3}^{(i,j)}$. Each associahedron-block has a subset of co-dimension one facets which are associated with $(m,n)\, \in\, \{\, (i,j),\, \dots,\, i + (j - i) - 1, j + (j - i) - 1)\, \}$ that are coloured red while all the other facets being coloured black.

We then show that the causal structure on ${\cal K}_{n}$ introduced in \cite{nima1912} has enough richness which we use to generalise the ABHY construction. We thus obtain convex realisation of an associahedron block $A_{n-3}^{(i,j)}$ in ${\cal K}_{n}^{+}$ and we denot this (ABHY) realisation as $A_{n-3}^{{\cal F}_{ij}}$. Here ${\cal F}_{ij}$ is the set of all the triangulations of an $n$-gon such that (i) Each triangulation has at most one  red diagonal from the set $\, \{\, (i,j), \dots,\, (i + \vert j - i\vert -1 , j + \vert j - i\vert\, - 1\, \}$. We then prove the following theorem.

Let $\Omega_{n}^{{\cal F}_{ij}}$ be the planar scattering form defined by the associahedron block $A_{n-3}^{(i,j)}$ with $\Omega_{n}^{{\cal F}_{ij}}(A_{n-3}^{ij})$ being its pull-back on the \emph{unique} convex realisation of the same associahedron block in ${\cal K}_{n}^{+}$. Similarly, let $\Omega_{n}^{\phi^{3}}$ be the planar scattering form defined in \cite{abhy1711} with $\Omega_{n}^{\phi^{3}}(A_{n-3})$ being its pull-back on a ABHY realisation.

The scattering form defined below defines the $n$-point amplitude of our theory upto sub-leading order in $\lambda_{2}$.
\begin{flalign}\label{master0}
\Omega_{n}^{Y}\, :=\, \left[\, \sum_{\vert i  - j \vert = 2}^{\frac{n}{2}}\, \frac{1}{\vert i - j \vert}\, \sum_{{\cal F}_{ij}}\, \Omega_{n}^{{\cal F}_{ij}}(A_{n-3}^{ij}) -\, \gamma\, \Omega_{n}^{\phi^{3}}(A_{n-3}) \right]
\end{flalign}
where 
\begin{flalign}
\gamma\, =\, \sum_{\textrm{Sum over all diagonals, (i,j)}}\, \frac{1}{\vert i - j\vert}\, -\, (n-3)
\end{flalign}
Hence, the following (sum over) $d-\log$ forms,
\begin{flalign}\label{master1}
\omega_{n}\, =\, \lambda_{1}^{n-2}\, \Omega_{n}^{\phi^{3}}(A_{n-3})\, +\, \lambda_{1}^{n-4}\, \lambda_{2}^{2}\, \Omega_{n}^{Y}
\end{flalign}
generates the tree-level planar amplitude up to order $\lambda_{2}^{2}$.

We then show that the effect of integrating out the massive $\phi_{2}$ field is geometrized in the world of positive geometries. It amounts to moving the associahedra to ``infinity" in different directions where they degenerate into a family of lower dimensional accordiohedra.

The paper is organised as follows. In section \ref{pgpsi}, we review the analysis of tree-level amplitudes for color ordered massless scalar theories with polynomial interactions . Although the main ingredients in this section are a review, we give an explicit evaluation of generic tree-level amplitudes generated by polynomial (massless) scalar interaction in terms of canonical forms in section \ref{asdln}. In sections \ref{bsft} and \ref{fgdt}, we introduce a two-scalar field theory with Lagrangian defined in equation \eqref{L1} and analyse the triangulations dual to Feynman diagrams. We argue how the naive attempt at using these dual triangulations to generate positive geometries such as associahedron fail if the number of particles $n\, \geq\, 6$.

In section \ref{ccdcp}, we use the remarkable causal structure in the  planar kinematic space introduced in \cite{nima1912} to locate convex polytopes  whose canonical forms generate tree-level perturbative amplitude of the two-scalar theory. As we prove in \ref{brcd}, combinatorially all the polytopes are in fact associahedra such that a unique linear combination over $d\log$ forms associated to their convex realisation is the scattering amplitude of interest. We show how these associahedra have boundaries some of which correspond to massless poles and some to the massive ones. In section \ref{sfeft}, we consider a low energy limit of these associahedra and show that they in fact project onto accordiohedra which are known to be positive geometries for the effective field theory that arises once we integrate out the massive field. We end with a discussion of some of the immediate open questions. 
\section{Positive Geometries for polynomial scalar interactions.}\label{pgpsi}
In this section, we review the ``Amplituhedron program" in the context of tree-level (and colour-ordered) scattering amplitudes of massless scalar theories with generic polynomial interaction.

We first begin by quickly reviewing the positive geometries, more specifically a class of simple polytopes known as accordiohedra (a polytope is a bounded, convex, higher dimensional generalisation of a polygon) which generate amplitudes for monomial ($\phi^{p}$.) interactions. A polytope ${\cal P}$ is called simple if each of the vertex is adjacent to exactly $\textrm{dim}({\cal P})$ co-dimension one boundaries called facets. Accordiohedron is a simple polytope whose co-dimension $k$ faces are in bijection with a set of dissections of polygon. Depending on the nature of dissections, that is, if the polygon is dissected into triangles or $p\, >\, 3$-gons, we get distinct accordiohedra. If the co-dimension $k$ faces of the simple polytope are in bijection with $k$-partial triangulation of an $n$-gon then the simple polytope is the $n-3$ dimensional associahedron $A_{n-3}$.\footnote{The precise definition of positive geometry is not required in this work but can be found in \cite{nima1703}. For our purpose, we may define positive geometry $X_{\geq\, 0}$ as  (i) a closed oriented subset in a projective space $X$ which has boundaries of all co-dimensions (ii) there is a unique canonical form on $X$ which has simple poles on (and only on) all the faces of $X_{\geq\, 0}$ and (iii) the residue of the canonical form is itself the canonical form defined intrinsically on the boundary as a positive geometry. For the purpose of this paper, we will need only a specific class of positive geometries, namely convex polytopes embedded in positive region of planar kinematic space ${\cal K}_{n}$.}

We then prove that a weighted sum over canonical forms associated to a family of simple polytopes which are closed under factorisation and gluing generate scattering amplitudes for a scalar field theories with polynomial interactions. Although the proof is simply a consequence of results proved in the literature, we take this opportunity to write down in complete generality the relationship between tree-level S-matrix for scalar theories and  a set of positive geometries. Reader not interested in this level of generality of already known results is encouraged to skip the proof in the first reading.\footnote{We thank Ashoke Sen for discussions on this issue and pressing on us the need to clarify the relationship between positive geometries and polynomial scalar interactions.}

We also caution the reader that the review is not self-contained  as we assume familiarity with the basic notions of positive geometry, especially associahedron, accordiohedron, their convex realisations and corresponding canonical forms. Interested reader is encouraged to read the original references, especially \cite{abhy1711, pinni1810, mrunmay3, mrunmay1911, songmatter} or recent review \cite{ferro}.

As was shown in \cite{pppp, prashanth, mrunmay3}, the tree-level planar (colour-ordered) amplitudes for massless $\phi^{p}$ interactions are sums over canonical forms of kinematic space accordiohedra. For quartic interaction, an $n$-particle amplitude is simply a (weighted) sum over canonical top forms associated to the accordiohedra of dimension $\frac{n-4}{2}$ \footnote{These accordiohedra were discovered by Baryshnikov and are known as Stokes polytopes \cite{baryshnikov}.}.  
Although explicit formulae have been derived for computing tree-level planar amplitudes associated to monomial-scalar interactions, an explicit classification of positive geometries whose canonical forms generate the scattering amplitude for generic polynomial interaction has never been analysed. In this section, we  fill this gap. Although the essential ingredients are simply review of known results, the final result has never been written down explicitly to the best of our knowledge.

This result essentially classifies the families of accordiohedra whose canonical forms generate amplitudes of a local unitary quantum field theory. We refer to such a set as a ``closed family". A closed family is defined as family of polytopes whose canonical forms generate amplitude of local unitary quantum field theories. We will see how different families of accordiohedra give us the amplitudes for massless scalar field theories. In the first section we will define accordiohedron and the families of accordiohedra; and in the second section we will discuss how to obtain a tree level scalar field theory amplitude from the families of accordiohedra.

\subsection{A Closed Family of Accordiohedra}\label{cfa}
Accordiohedra are generalization of a family of polytopes called associahedra. The $n-3$ dimensional associahedron $\mathcal{A}_{n-3}$ is a simple polytope (i.e., each vertex is adjacent to $n-3$ co-dimension one faces) whose vertices are in one-to-one correspondence with triangulations of an  $n$-gon and whose facets (co-dimension one boundaries) are in one-to-one correspondence with diagonals of $n$-gon. The facets of associahedra satisfy the following property which we call ``factorization''. Any $d$-dimensional facet of $\mathcal{A}_{n-3}$ is a product of two lower dimensional associahedra $\mathcal{A}_{r} \times \mathcal{A}_{d-r} $ for some $r$ ($0\leq r \leq d$). 

Accordiohedron is a generalization of associahedra where the vertices of the polytope are associated with various dissections of $n$-gon instead of just triangulations. For example, we could consider a polytope where the vertices are associated with quagrangulations of an $n$-gon. Such a polytope is known as Stokes polytope. It can be easily verified with simple example of a hexagon that we can not have a simple polytope with vertices in one-to-one correspondence with quadrangulations of $n$-gon and facets (co-dimension one boundaries) in one-to-one correspondence with diagonals. There are more than necessary quandrangulations and not all diagonals can be part of a quadrangulation.

To deal with this we introduce a notion of ``compatiblity" with a reference quadrangulation. We say a diagonal $\delta$ of an $n$-gon is compatible with reference quadrangulation $Q$ if the set of diagonals and sides of $Q$ that intersect the segment $\delta^{\prime}$ obtained by small clock-wise rotation of $\delta$ is connected. The Stokes polytope associated with reference quadrangulation $Q$ is a simple polytope whose vertices are in one-to-one correspondence with quadrangulations formed by diagonals compatible with $Q$ and whose co-dimension one facets are in one-to-one correspondence with diagonals compatible with $Q$. We should note here that, unlike the family of associahedra there are more than one Stokes polytopes of a given dimension. We have one Stokes polytope for every quadrangulation. However, just like associahedra the facets of Stokes polytopes are products of lower dimensional Stokes polytope.

Now we define $\mathcal{AC}(D)$ the accordiohedron associated with a reference dissection $D$. To do that we first have to define the notion of compatibility of a diagonal with a reference dissection just as defined it for quadrangulations. We say a diagonal $\delta$ of an $n$-gon is compatible with reference dissection $D$ if the set of diagonals and sides of $D$ that intersect the segment $\delta^{\prime}$ obtained by small clock-wise rotation of $\delta$ is connected.(See figure \ref{acccompatibility}).
\begin{figure}
    \centering
    \includegraphics[scale=0.5]{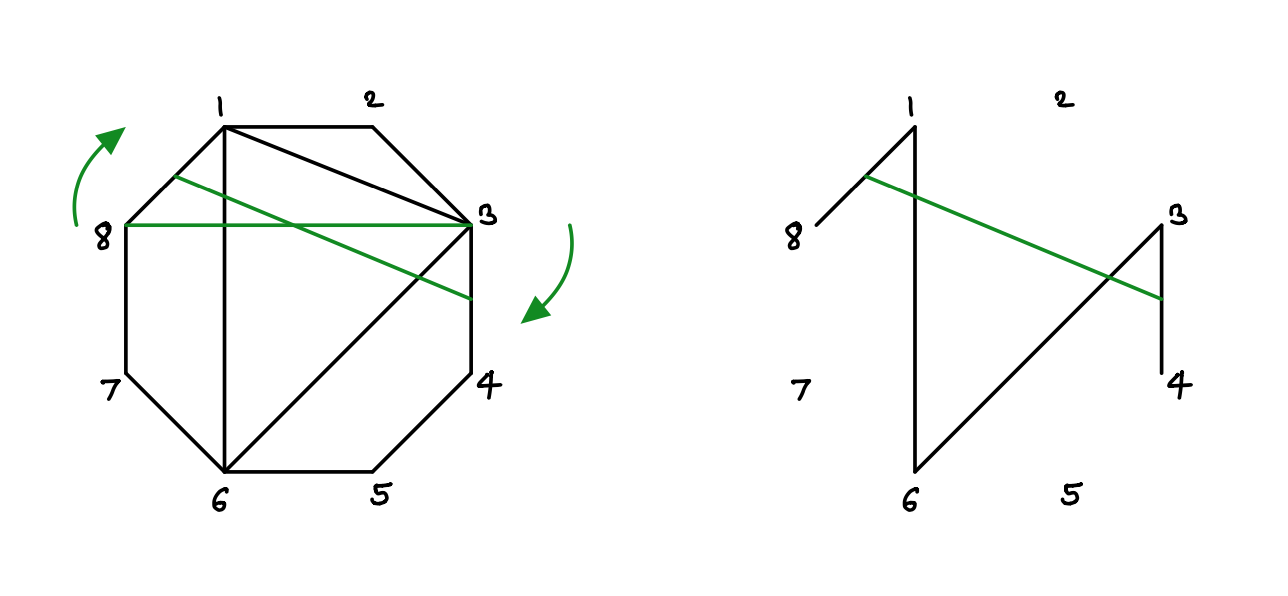}
    \caption{The diagonal (38) of the (dual) octagon is compatible with the reference \{13,36,16\}}
    \label{acccompatibility}
\end{figure}

The accordiohedron associated with dissection $D$ is a simple polytope whose vertices are in one-to-one correspondence with maximal dissections formed by diagonals compatible with $D$ and whose co-dimension one facets are in one-to-one correspondence with diagonals compatible with $D$.

We have one accordiohedra for every dissection of a polygon. Thus the space of accordiohedra is labeled by dissections of polygons. Associahedra and Stokes polytopes are some examples of accordiohedra where the reference dissections are triangulations and quadrangulations respectively. Now we will look at different ways of classifying the space of accordiohedra.

To begin with, we can classify the space of accordiohedra into different families using the following criteria
\begin{itemize}
    \item Number of sides (external sides) of the reference dissection
    \item Dimension of accordiohedra, or equivalently number of diagonals in the reference dissections.
    \item Constituents or building blocks of dissection (triangles, quadrilaterals, pentagons, etc.)
\end{itemize}
Note, these are the criteria for reference dissection, which is expected as dissections label the space of accordiohedra. All these criteria are encoded in the following infinite dimensional vector associated with the dissection, $\vb{v}=(v_{3},v_{4},v_{5},\ldots)$. Where $v_{r}$ is the number of $r$-gons in the reference dissection. Hence, all but finitely many components of $\vb{v}$ are zero. Number of sides of reference dissection is given by 
\be 
n = 2 + \sum_{i=3}^{\infty} (i-2) v_{i}. 
\ee 
The dimension of accordiohedra is given by 
\be 
\textrm{Dim}\left[\mathcal{AC}(D)\right] = -1 +   \sum_{i=3}^{\infty}  v_{i}. 
\ee 
The family of associahedra consists of accordiohedra with $\vb{v} = (n-2,0,0,\ldots)$ and the family of Stokes polytopes consists of accordiohedra with $\vb{v}=(0,\frac{n-2}{2},0,0,\ldots)$.

There are two further classification of the space of accordiohedra which will play a role in identifying the family of polytopes whose forms generate tree-level amplitudes of massless scalars. One of these classifications is coarser than the three discussed above and the other one is finer.

{\bf Coarser classification} :\\
The family of associahedra and Stokes polytopes satisfy the interesting property that facets of associahedra are products of lower dimensional associahedra and facets of Stokes polytopes is product of lower dimensional Stokes polytopes.
We would like to classify the set of accordiohedra which are closed under such factorisation.

We first define pure accordiohedron. Any accordiohedron is called pure if the reference dissection breaks the polygon into a $p$-cells with $p\, \geq\, 3$. Pure accordiohedra are interesting as each tower of Pure accordiohedra ${\cal AC}(D)$ (where $D$ is a $p$-gulation of an $n$-gon and tower is defined with respect to the number of external vertices $n$)  is closed under factorization. 


 Now we introduce the notion of ``closed under gluing". We say a set of accordiohedra $S$ is closed under gluing if taking the reference dissections of any two elements of $S$ and gluing them along some side gives you the reference of some element of $S$ then we say $S$ is closed under gluing. The sets of associahedra and stokes polytopes are examples of sets closed under gluing. 
 
 If we take two dissections, $D_{1}$ of $n_{1}$-gon and $D_{2}$ of $n_{2}$-gon, and glue these dissections along some side, we get a dissection $D$ of $n_{1}+n_{2}-2$-gon. The side along which the dissections $D_{1}$ and $D_{2}$ were glued becomes a diagonal of $D$. Let's denote that diagonal by $\delta$. The vectors $\vb{v}_{1}$, $\vb{v}_{2}$, and $\vb{v}$ associated with $D_{1}$, $D_{2}$, and $D$, respectively, are related by $\vb{v}_{1} + \vb{v}_{2} = \vb{v} $. Thus, the dimensions of $\mathcal{AC}(D_{1})$, $\mathcal{AC}(D_{2})$ and $\mathcal{AC}(D)$ are related by $ \mathrm{Dim}[\mathcal{AC}(D_{1})] + \mathrm{Dim}[\mathcal{AC}(D_{2})] = \mathrm{Dim}[\mathcal{AC}(D)] -1 $. Further, the co-dimension one facet associated with the diagonal $\delta$ is given by $\mathcal{AC}(D_{1}) \times \mathcal{AC}(D_{2})$. That is, the facet $\delta$ of $\mathcal{AC}(D)$ factorizes into $\mathcal{AC}(D_{1}) \times \mathcal{AC}(D_{2})$. Therefore, the notion of closed under gluing is closely related to the notion of closed under factorization. However they are not the same. An example of set closed under factorization but not closed under gluing is the set of associahedra with dimension less than 10. An example of a set closed under gluing but not closed under factorization is the set of associahedra with dimensions greater than 9. 
 
 The coarser families of accordiohedra we are interested in are the sets of accordiohedra which are closed under factorization and closed under gluing.  An example of set closed under factorization and closed under gluing is the set of accordiohedra whose vector $\vb{v}$ is of the type $(v_{1},v_{2},0,0,\ldots)$. This set contains all associahedra and all Stokes polytopes along with all accordiohedra where the reference dissection has both triangles and quadrilaterals.  

We claim that all sets which are closed under factorization and closed under gluing are of the type 
\be 
S_{i_{1},i_{2},\ldots,i_{r}} = \{ \mathcal{AC}(D) | \vb{v}_{D} = \sum_{j=1}^{r} a_{i_{j}} \hat{e}_{i_{j}} \}.
\ee 
Where $\hat{e}_{i_{j}}$ is a unit vector with one in $i_{j}-2$ th position and zero in other positions. So the set of all associahedra is denoted by $S_{3}$, and the set of all associahedra and all Stokes polytope along with all accordiohedra where the reference dissection has both triangles and quadrilaterals is denoted by $S_{3,4}$. 


{\bf Finer classification}:\\
Now we move on to the second, finer classification. The Dihedral group $D_{n}$ acts on the set of all dissections of $n$-gon. The the action of $D_{n}$ dissection does not change the building blocks of the dissection. That is $\vb{v}_{D}= \vb{v}_{g\cdot D}$. Where, $g\cdot D$ denotes the action of $g \in D_{n}$ on $D$ and $\vb{v}_{D}$ and $\vb{v}_{g\cdot D}$ are the infinite dimensional vectors associated with disections $D$ and $g\cdot D$ respectively. Thus, there is an action of $D_{n}$ on all dissections with a given $\vb{v}$. We can classify the set of all dissections with a given $\vb{v}$ by the orbits of this action. We denote the set of orbits by $\mathcal{P}_{\vb{v}}$ and call the representative dissection of a orbit as the ``primitive" of that orbit.   

\subsection{Amplitude as a Sum Over d-$\log$ Forms}\label{asdln}
A set of accordiohedra which is closed under factorization and closed under gluing gives us tree level planar scalar field theory amplitudes. In this section, we will describe how we get tree level planar scalar field theory amplitudes from $S_{i_{1},i_{2},\ldots, i_{r}}$.

The $n$-point amplitude is obtained from the subset $S_{i_{1},i_{2},\ldots, i_{r}}^{(n)}$ of $S_{i_{1},i_{2},\ldots, i_{r}}$ where $\sum_{j=1}^{r} (i_{j}-2) a_{i_{j}} = n-2 $ that is number of sides of reference dissections is n. Suppose 
\be 
\mathcal{V}^{(n)}_{i_{1},i_{2},\ldots, i_{r}} = \{\vb{v} |\vb{v}  = \sum_{j=1}^{r} a_{i_{j}} \hat{e}_{i_{j}}, \text{ with } \sum_{j=1}^{r} (i_{j}-2) a_{i_{j}} = n-2  \},
\ee 
for ease of notation we will drop $i_{1},i_{2},\ldots, i_{r} $  in the subscript and just write $\mathcal{V}_{n}$. 

We will soon see that a function $\omega(\mathcal{AC}(D)) : \mathcal{K}_{n} \rightarrow \mathbb{R} $ is associated with accordiohedra $\mathcal{AC}(D)$. Where $\mathcal{K}_{n}$ is the $n$-particle kinematics space. The scattering amplitude is a weighted sum of these functions over the space $S_{i_{1},i_{2},\ldots, i_{r}}^{(n)}$.  
\be 
\mathcal{M}_{n}(p_{1},\, p_{2},\, \ldots, p_{n}) = \sum_{\vb{v}\in \mathcal{V}_{n}}   \lambda_{i_{1}}^{a_{i_{1}}} \cdots \lambda_{i_{r}}^{a_{i_{r}}}  \sum_{D \in \mathcal{P}_{\vb{v}}}  \hspace{0.25cm} \sum_{\sigma \in D_{n} / G_{D} }\alpha_{\sigma \cdot D} \hspace{0.1cm} \omega(\sigma \cdot D).
\ee
Where $G_{D} = \{ g \in D_{n}| g\cdot D = D \}$ is the stabiliser of $D$ under the action of dihedral group. It can be shown that we can choose the weights $\alpha_{D}$ to be equal for $D$s belonging to the same orbit under the action of dihedral group. For a given $\vb{v}$ the term $ \sum_{D \in \mathcal{P}_{\vb{v}}}  \hspace{0.25cm} \sum_{\sigma \in D_{n} / G_{D} }\alpha_{\sigma \cdot D} \hspace{0.1cm} \omega(\sigma \cdot D)$ is independent of the family of accordiohedra $S_{i_{1},i_{2},\ldots, i_{r}}^{(n)}$.

The combinatorial polytope $\mathcal{AC}(D)$ gives us the canonical form associated with it on the kinematic space. This form is given by, 
\be 
\Omega\left[\mathcal{AC}(D) \right] = \sum_{D_{a} \in \mathcal{AC}(D) }  \mathrm{Sgn}(D_{a}) \bigwedge_{(ij)\in D_{a} } d \log(X_{ij}).
\ee
Where $X_{ij}= (p_{i}+ p_{i+1}+ \cdots+ p_{j-1})^{2}$ are variables on the kinematic space. Now, to get the function $\omega(\mathcal{AC}(D))$, we have to realise the combinatorial polytope $\mathcal{AC}(D)$ in the kinematic space and restrict the canonical form on the kinematic space accordiohedra. 
\be 
 \Omega\left[\mathcal{AC}(D) \right]\vert_{\mathcal{AC}(D)_{k}} = \omega(\mathcal{AC}(D) ) \bigwedge_{(ij) \in D } d X_{ij}.
\ee
It can be shown that $\omega(D) = \sum_{D_{i}\in \mathcal{AC}(D)} \psi(D_{i})$ where $\psi(D_{i}) =  \prod_{(ab)\in D_{i}} \frac{1}{X_{ab}}$. 

Now let's see how we determine the weights $\alpha$. We can write the amplitude as  
\begin{align}
\mathcal{M}_{n}(p_{1},\, p_{2},\, \ldots,\, p_{n}) &= \sum_{\vb{v}\in \mathcal{V}_{n}}   \lambda_{3}^{v_3} \cdots \lambda_{n}^{v_n}  \sum_{D \in \mathcal{P}_{\vb{v}}}  \alpha_{D} \sum_{\sigma \in D_{n} / G_{D} }            \sum_{D_{i}\in \mathcal{AC}(\sigma \cdot D)} \psi(D_{i}) .
\end{align}
We fix the weights demanding that all poles of $\mathcal{M}(1,2,\ldots, n) $ come with residue one.

To analyse the requirement on the weights it would be useful to define  $\delta(D_{i},D_{j})$, $ N([D_{i}],[D_{j}]) $, $M([D_{i}],[D_{j}])$, and $\Psi([D])$ as follows
\begin{itemize}
\item $\delta(D_{i},D_{j})$ tells you wether the dissection $D_{j}$ occurs in $\mathcal{AC}(D_{i})$ or not. That is, 
\be
\delta(D_{i},D_{j})=
 \begin{cases} 
1 \text{ if $D_{j} \in \mathcal{AC}(D_{i})$ }\\
0 \text{ if $D_{j} \notin \mathcal{AC}(D_{i})$ }
\end{cases}
\ee
For any $\sigma \in D_{n}$,
\begin{align}
\delta(D_{i},D_{j}) = \delta(\sigma \cdot D_{i},\sigma \cdot  D_{j})
\end{align}
\item $N([D_{i}],[D_{j}])$ is the number of times any dissection $D_{k}$ in the orbit of $D_{j}$ occurs in the set of all accordiohedra of dissections in the orbit of $D_{i}$. That is,
\be 
N([D_{i}],[D_{j}]) = \sum_{D_{\ell}\in [D_{i}]} \delta(D_{\ell},D_{k}) \hspace{2cm} \text{for any } D_{k}\in [D_{j}] .
\ee
\item $M([D_{i}],[D_{j}])$ is the number of elements from the orbit of $D_{i}$ occur in any accordiohedron of a dissection in the orbit of $D_{j}$. That is,
\be 
M([D_{i}],[D_{j}]) = \sum_{D_{k}\in [D_{i}]} \delta(D_{j},D_{k}).
\ee
They are related as follows,
\be 
N([D_{i}],[D_{j}]) = \frac{| G_{D_{j}} |}{| G_{D_{i}} |}M([D_{j}],[D_{i}]).
\ee
Suppose $[D_{i}] \in \mathcal{P}_{\vb{v}_{i}}$ and $[D_{j}] \in \mathcal{P}_{\vb{v}_{j}}$, and if $\vb{v}_{i} \neq \vb{v}_{j}$ then $\delta(D_{i},D_{j}) =0 $  and hence $N([D_{i}],[D_{j}])=M([D_{i}],[D_{j}]) =0$.
\item Lastly,
\be 
\Psi([D]) = \sum_{\sigma \in D_{n}/G_{D}} \psi(\sigma \cdot D).
\ee
\end{itemize}
Now the scattering amplitude can be expressed as
\be 
\mathcal{M}_{n}(p_{1},\, p_{2},\, \ldots, p_{n}) = \sum_{\vb{v}\in \mathcal{V}_{n}}   \lambda_{3}^{v_3} \cdots \lambda_{n}^{v_n} \sum_{[D_{i}]\in \mathcal{P}_{\vb{v}}}  \sum_{[D_{j}]\in \mathcal{P}_{\vb{v}}}  \alpha_{[D_{j}]} N([D_{j}],[D_{i}])  \Psi([D_{i}]). 
\ee
The requirement on the weights is equivalent to the following system of linear equations
\be 
 \sum_{[D_{j}]\in \mathcal{P}_{\vb{v}}}  \alpha_{[D_{j}]} N([D_{j}],[D_{i}])  = 1 \hspace{1.5cm} \forall [D_{i}] \in \mathcal{P}_{\vb{v}}.
\ee

All the vertices (triangulations) of combinatorial associahedra are on equal footing, whereas in other accordiohedra the reference dissection is special compared to the rest. This disparity is captured in the weight of the accordiohera.

We end this section by briefly reviewing the positive geometry for massive bi-adjoint scalar amplitude. Although our work is concerned with scattering amplitude of massless particles (but with massive as well as massless propagators), we review the known extension of the associahedron program to massive scalar amplitudes.

The positive geometry for tree-level colour ordered amplitude of a massive bi-adjoint scalar is simply the ABHY associahedron whose facets are located $X_{ij}\, =\, m^{2}$ \cite{nima-talk}. The planar kinematic space of massive particles is defined as follows. Consider the massive kinematic space ${\cal K}_{n}^{(m)}$ co-ordinatized by the planar variables, 
\begin{flalign}
\tilde{X}_{ij}\, =\, (p_{i}\, +\, \dots,\, +\, p_{j-1})^{2}\, -\, m^{2}
\end{flalign}
Then $\tilde{X}_{i,i+1}\, =\, \tilde{X}_{1n}\, =\, 0$.

Just as the massless kinematic space, ${\cal K}_{n}^{(m)}$ is $\frac{n(n-3)}{2}$ dimensional and we can locate ABHY associahedra in the positive region ${\cal K}_{n}^{m +}$ of the kinematic space. The corresponding canonical form equals the tree-level planar amplitude for the massive bi-adjoint $\phi^{3}$ theory. 

\section{Two-Scalar Field Theory with Cubic Interaction}\label{bsft}
As stated in the introduction, our aim is to analyse tree-level scattering amplitudes in a theory with two scalar fields $\phi_{1}^{a A},\, \phi_{2}^{b B}$ that transform in the bi-adjoint representation of $U(N)\, \times\, U(\tilde{N})$.
The Lagrangian is
\begin{flalign}\label{L1}
\begin{array}{lll}
L(\, \phi_{1},\, \phi_{2})&=\\
&\frac{1}{2}\, \sum_{i=1}^{2}\, \partial_{\mu}\, \phi_{i aA}\, \partial^{\mu}\, \phi_{i}^{aA}
\, -\, \frac{1}{2} m^{2}\phi_{2}^{aA}\, \phi_{2\, aA}\, -\, \frac{\lambda_{1}}{3!}\, f_{abc}\, \tilde{f}_{ABC}\, \phi_{1}^{aA}\phi_{1}^{bB}\, \phi_{1}^{cC}\\
&\hspace*{2.9in}-\, \frac{\lambda_{2}}{2!}\, f_{abc}\, \tilde{f}_{ABC}\, \phi_{1}^{aA}\, \phi_{1}^{bB}\, \phi_{2}^{cC}
\end{array}
\end{flalign}
Tree-level amplitudes in a multi-scalar field theory with Yukawa-type interaction between the bi-adjoint scalar field and a ``massless" matter field which transforms in (bi)fundamental representation of $U(N)\, \times\, U(\tilde{N})$ have been analysed in a beautiful series of works \cite{songmatter,fei}. The bi-adjoint scalar field self-interaction as well as the Yukawa interactions were taken to have the same coupling constant $\lambda$ in these works. It was proved that the positive geometries underlying amplitudes in such a theory are what are called open associahedra.

While certainly motivated by the analysis in \cite{songmatter, fei}, we have a different set up involving one massive and one massless scalar field with two independent coupling constants which govern the $\phi_{1}^{3}$  and $\phi_{1}^{2}\phi_{2}$ interactions respectively.

Schematically, we are interested in amplitude generated by interaction potential,
\begin{flalign}\label{pot1}
V(\phi_{1}, \phi_{2})\, =\, (\lambda_{1} \phi_{1}^{3}\, +\, \lambda_{2} \phi_{1}^{2}\, \phi_{2}\, )
\end{flalign}
where $\phi_{1}, \phi_{2}$ are massless and massive scalar fields respectively.

We will be specifically interested in tree-level and colour-ordered (CO) amplitudes involving \emph{only} massless external particles. Much of the analysis in our paper can be generalised for generic configurations of the external particles, but as we are eventually interested in integrating out the massive modes, our primary focus is on amplitudes involving massless external states. We will denote such amplitudes as ${\cal M}_{n}^{\textrm{CO}}(p_{1}, \dots,\, p_{n})$ with $p_{i}^{2}\, =\, 0\, \forall\, i$. Perturbative evaluation of the amplitude to sub-leading order in $\lambda_{2}$ can be written as,
\begin{flalign}\label{march8-1}
{\cal M}^{\textrm{CO}}_{n}(p_{1}, \dots, p_{n})\, =\, \lambda_{1}^{n-3}\, {\cal M}_{n}^{CO\, (1)}(p_{1}, \dots, p_{n})\, +\, \lambda_{2}^{2}\, \lambda_{1}^{n-4}\, {\cal M}_{n}^{\textrm{CO} (2)}(p_{1}, \dots, p_{n})
\end{flalign}
${\cal M}_{n}^{\textrm{CO}\, (1)}$ is generated by $\phi_{1}^{3}$ vertices in which all the propagators are massless and ${\cal M}_{n}^{\textrm{CO}\, (2)}$ is generated by all but two $\phi_{1}^{2}\, \phi_{2}$ vertices with one massive and $n-4$ massless propagators. Thus the set of all the Feynman graphs that contribute to the total amplitude ${\cal M}_{n}^{\textrm{CO}}$  has at most one  $\phi_{2}$ propagator.

The effective field theory involving only the massless scalar field  is obtained by taking $m\, \rightarrow\, \infty,\, \lambda_{2}\, \rightarrow\, \infty$ such that $\frac{\lambda_{2}}{m}\, =\, g\, =\, \textrm{fixed}$. This ensures that in the lower energy limit we obtain amplitudes associated to
\begin{flalign}\label{pot2}
V_{\textrm{eff}}(\phi_{1})\, =\, \lambda_{1}\, \phi_{1}^{3}\, +\, g\, \phi_{1}^{4}
\end{flalign}
Our goal is to find positive geometries in planar kinematic space ${\cal K}_{n}$  whose set of co-dimension one facets either correspond to a pole of the massive propagator or the massless propagator. For the amplitudes of interest (namely up to order $\lambda_{2}^{2}$) each vertex of such a positive geometry can be adjacent to at most one ``massive facet".

\subsection{From Feynman Graphs to Dual Triangulations}\label{fgdt}
In the case of the bi-adjoint cubic scalar coupling with a single massless field, the  positive geometries are convex realisations of the combinatorial associahedron, each of whose vertices are in bijection with complete  ``mono-chromatic triangulations" of an $n$-gon where each diagonal is dual to a propagator in the Feynman graph.

Hence we first need to classify triangulations which are dual to the Feynman graphs that produce amplitude defined in equation \eqref{march8-1}. The set of Feynman graphs that produce the amplitude proportional to $\lambda_{1}^{n-3}$ is dual to triangulations of an $n$-gon where the diagonals have no additional labels. As ABHY has taught us, it is the  immensely deep structure hidden in the combinatorics of these traingulations which generate positive geometry for a single scalar field amplitude. 

However, for the amplitude contribution proportional to $\lambda_{2}^{2}$, such mono-chromatic triangulations are not appropriate as the contributing Feynman graphs have precisely one $\phi_{2}$ propagator. Hence the dual triangulations are such that precisely one of the diagonals in a \emph{complete} triangulation is distinguished from the other diagonals. We colour such diagonals red  to differentiate them from the black diagonals.  We refer to triangulations in which all but one diagonals are black as ``mostly black  triangulations" ( See figure \ref{offcolor}). 


\begin{figure}[H]
    \centering
    \includegraphics[scale=0.5]{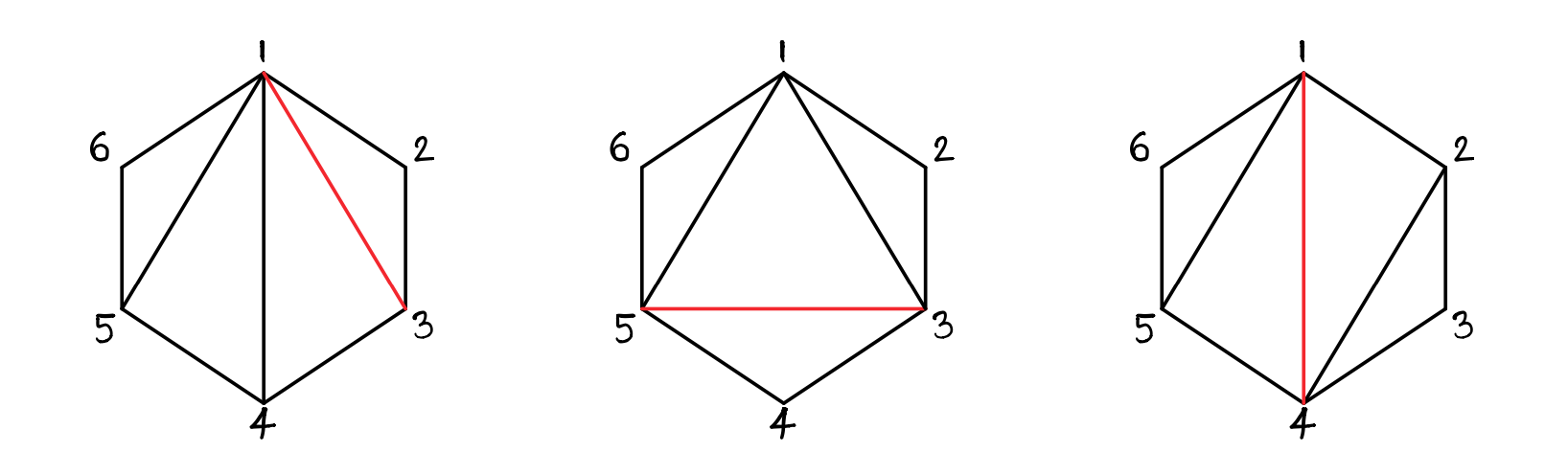}
    \caption{Mostly Black triangulations with (13) red and (14),(15) black on the left, (35) red and (13),(15) black in the center and (14) red and (24), (15) black on the right. }
    \label{offcolor}
\end{figure}
It is immediately obvious that mostly black  triangulations are in one to one correspondence with cubic Feynman graphs where all but one propagators are massless and one of the propagator is massive.

The simplest mutation rule one can define on the mostly black  triangulations is the one where mutation does not change colour of the diagonal. We can then immediately see that mostly black  triangulations produce  combinatorial polytopes in one and two dimensions . We denote a black diagonal between vertices $i$ and $j$ as $(i,j)$ and the red diagonal as $\textcolor{red}{(k,l)}$.\footnote{When we write a set of diagonals, we will indicate each diagonal as $(ij)$ instead of explicitly writing it out as $(i,j)$.}
\begin{itemize}
\item Triangulation of a $4$-gon with precisely one red diagonal and with the mutation rule which does not change colour produces one dimensional associahedron with the vertices being labelled by two red-diagonals $\textcolor{red}{(1,3)}$ and $\textcolor{red}{(2,4)}$.
\item  Mostly black  triangulation of a $5$-gon produces a two dimensional polygon with ten edges, whose vertex set is in bijection with $ ( (i,j),\, \textcolor{red}{(k,l)})$ 
\end{itemize}
Let us first start with $n\, =\, 4$ case. In this case the requirement of mostly black triangulations which has at least one red diagonal has precisely one red (and no black diagonal) and the resulting polytope is simply the one dimensional assocaihedron $A_{1}$.

In the case of $n\, =\, 5$ the resulting polytope is a $10$-gon and is known in the literature as the two dimensional colorful associahedron $A_{2}^{c}$. \cite{colorass1}.
An $n$-dimensional colorful associahedron $A_{n-3}^{c}$ is combinatorial polytope associated with colored triangulations of $n+3$-gons in which the diagonals are assigned a color from a set of $n$ colors. 
$A_{2}^{c}$ is shown in the figure \ref{colorfulasso} below. However these simple mutation rules do not work for two reasons.\\
\begin{figure}[H]
    \centering
    \includegraphics[scale=0.33]{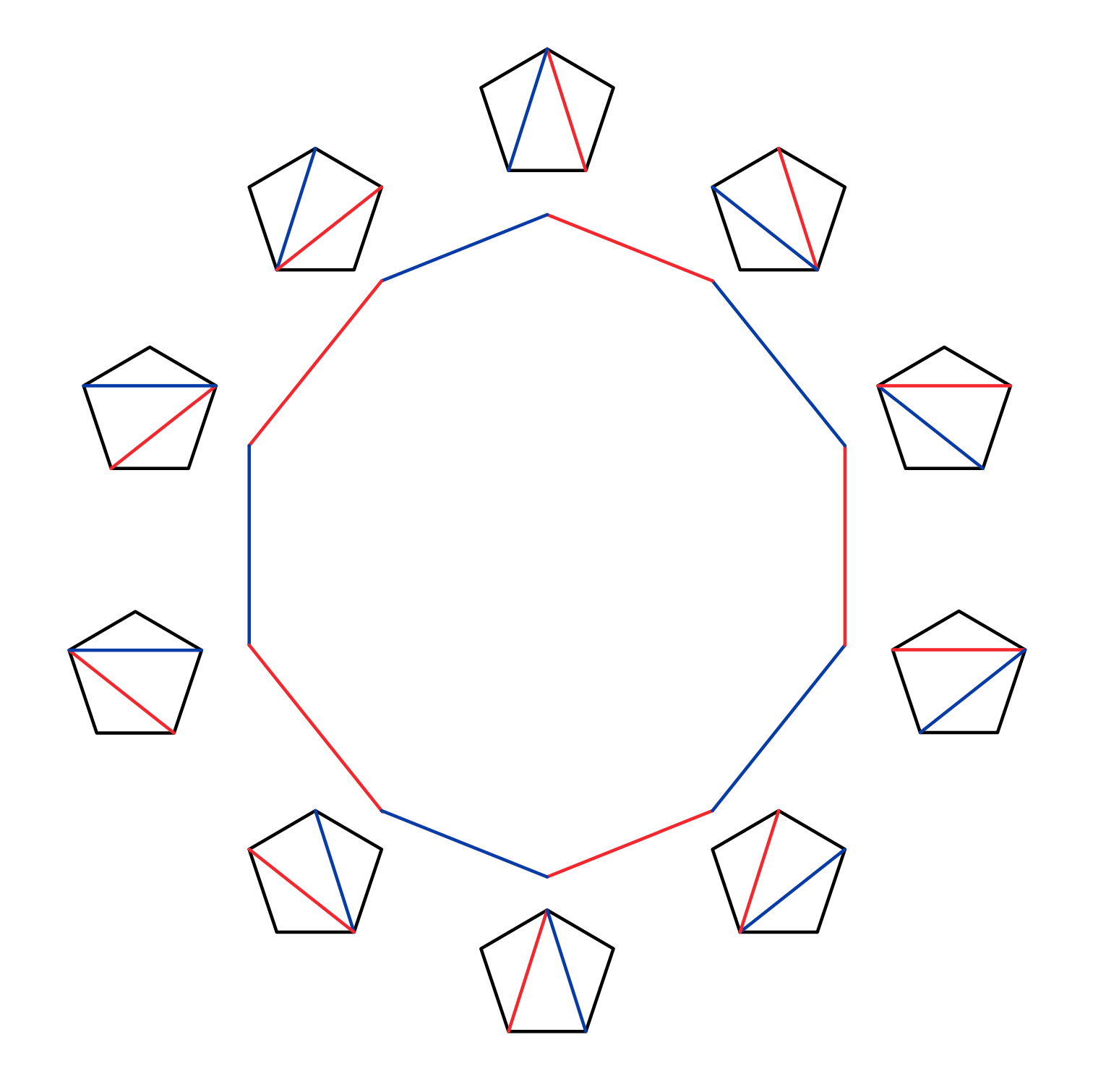}
    \caption{$A_{2}^{c}$}
    \label{colorfulasso}
\end{figure}

\begin{itemize}
\item If we consider complete mostly black triangulations of an $n$-gon where $n\, \geq\, 6$ then there are no closed combinatorial polytopes.  This can be seen by looking at the example of  $n\, =\, 6$.  Let us assume  on the contrary  that there exists an polytope $A_{3}^{r}$ whose vertices are in bijection with complete triangulation of the hexagon in which precisely one diagonal is red. 
If ${\cal A}_{3}^{r}$ exists, then each of it's co-dimension one facets are in bijection with one partial triangulation of the hexagon. Hence we can have facets corresponding to partial triangulation where a black diagonal $(i,j)$ or a red diagonal $\textcolor{red}{(i,j)}$ is deleted.
Clearly, each such facet is either a square (red diagonal with $ j - i\, =\, 3$ (modulo 6) ), pentagon (red diagonal with $ j - i\, =\, 2$ (modulo 6)) or a decagon (black diagonal with $ j - i\, =\, 2$ (modulo 6)). However there is no unique facet associated to black diagonal $(i,j)$ with $ j - i\, =\, 2$ (modulo 6). This can be seen as follows. Consider e.g. partial triangulation obtained by deleting the diagonal $\textcolor{red}{(1,4)}$. The corresponding facet is adjacent to Four Facets obtained by deleting the diagonals 
\begin{flalign}
\{\, (1,3),\, (2,4),\, (4,6),\, (1,5)\, \}\nonumber 
\end{flalign}
respectively.
However there is no unique facet associated to $(1,4)$ (and in general any $(i,j)$), as adjacent to such a facet,  there are two possible sets of vertices,
\begin{flalign}
\begin{array}{lll}
S_{1}\, =\, \{\, (\textcolor{red}{13}, 14, 15)\, (\textcolor{red}{24}, 14, 15),\, (\textcolor{red}{24}, 14, 46)\, (\textcolor{red}{13}, 14, 46)\, \}
\end{array}
\end{flalign}
and
\begin{flalign}
\begin{array}{lll}
S_{2}\, =\, \{\, (\textcolor{red}{15}, 14, 13)\, (\textcolor{red}{15}, 14, 24),\, (\textcolor{red}{46}, 14, 24)\, (\textcolor{red}{46}, 14, 13)\, \}
\end{array}
\end{flalign} 
Clearly if we choose one of the $S_{1}$ or $S_{2}$ as being the set of vertices belonging to the $(1,3)$-facet then we immediately see that $(1,3)$ and $(1,5)$ which are adjacent in the neighbourhood of $\textcolor{red}{(1,4)}$ are not adjacent in the neighbourhood of the facet associated to $(1,4)$. A moment of meditation will convince the reader that this example clearly indicates why the set of all mostly-black triangulations can not form a closed polytope.\footnote{We thank Vincent Pilaud for discussion on this issue.}
This is the first obstruction towards realising positive geometry for scattering amplitudes in a theory with $V(\phi_{1}, \phi_{2})\, =\, \phi_{1}^{2}\phi_{2}$ where all the external particles are massless.
\item The second obstruction to simple minded mutation rules  comes from the impossibility of convex realisation of colorful associahedron in ${\cal K}_{n}^{+}$.\\
\begin{lemma}
The two dimensional colorful associahedron $A_{2}^{c}$ which is the combinatorial polytope for $n, =\, 5$ particles can not be realised in ${\cal K}_{n}^{+}$.\nonumber
\end{lemma}
{\bf Proof} : 
We denote the massive facet $X_{ij}\, =\, m^{2}$ as $\tilde{X}_{ij}\, =\, 0$.
$A_{2}^{c}$ is a 10-gon obtained by unfolding $A_{2}$ by painting alternating faces red. (See figure \ref{colorfulasso}).  There is no planar polytopal realisation of $A_{2}^{c}$ such that $X_{13}\, =\, 0$ is parallel to $\tilde{X}_{13}\, =\, 0$, $\tilde{X}_{25}\, =\, 0$ and $X_{25}\, =\, 0$.\footnote{Any other realisation will not have the canonical form all whose residues vertices are $\pm 1$ and hence will not define any tree-level amplitude.} It can be immediately checked that  canonical form associated to any other realisation will not produce the tree-level amplitude for the bi-scalar theory.
\end{itemize}
However these obstructions in fact guides us in our pursuit of positive geometries  whose boundaries correspond to all the singularities of the desired  amplitude. 
In the next section, we show that although $A_{2}^{c}$ can not be realised as a convex polytope in the (positive region of) kinematic space, union of all the facets of $A_{2}^{c}$ and $A_{2}$ arrange themselves into a family combinatorial polytopes. 
\section{Coloured Causal Diamonds and Convex polytopes}\label{ccdcp}
\subsection{A Brief Review of Causal Diamonds}\label{brcd}
The convex realisation of associahedron discovered by ABHY in \cite{abhy1711} was formalised in terms of polytopal fans generated by type-A clusters in \cite{thomas-matte}. The ``naturalness" of these specific realisations arises from the fact that associahedra are naturally associated to type-A cluster algebra and the cluster algebra has enough structure to in turn produce convex realisations of $A_{n-3}$ in ${\cal K}_{n}^{+}$.

However, the connection and dependence on cluster algebra, though striking can be unnerving for a physicist and it would certainly be pleasing to know if the naturalness underlying ABHY realisations could be understood in terms of some simple ``physics" principles.
A remarkable progress in answering this question was made by Arkani-Hamed et. al. in \cite{nima1912} who discovered a  
``causal" structure on the planar kinematic space ${\cal K}_{n}$. The causal structure is realised via  interpreting $X_{ij}$ variables as discretization of a two dimensional massless scalar field $X(u,v)$ with $u$ and $v$ being (abstract) retarded and advanced co-ordinates on the two dimensional flat space-time.

More in detail, in \cite{nima1912}, the authors introduced a causal diamond in ${\cal K}_{n}$ which is a null lattice with vertices labelled by the planar kinematic variables.  $X_{ij}$s with fixed $i$ are placed along the $u\, =\, \textrm{const}$ ``null lines and cyclic symmetry of the $X$ variables, $X_{i, j+n}\, =\, X_{ij}$ can be broken by choosing a strip in the causal kinematic space such that $X_{i,i\pm 1}\, =\, 0$. Kinematic space causal diamonds are best understood by simply staring at the figure \ref{causaldiamond5} below.
\begin{figure}[ht]
    \centering
    \includegraphics[scale=0.4]{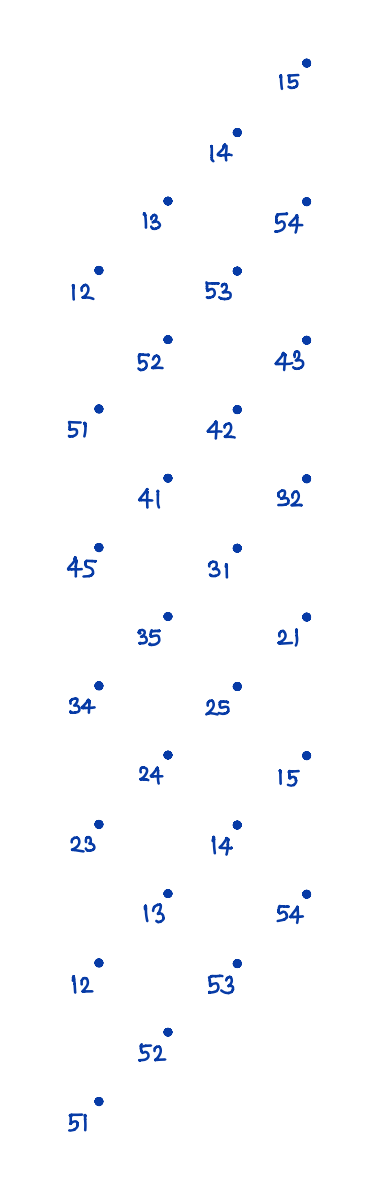}
    \caption{Causal diamond for $n=5$}
    \label{causaldiamond5}
\end{figure}
\pagebreak

Causality in kinematic space beautifully captures compatibility of the diagonals\footnote{Here by compatible diagonals we mean non-intersecting, or equivalently diagonals which together can be a part some triangulation.}. \cite{nima1912} : Given any diagonal $(i,j)$ of an $n$-gon which labels a planar variable $X_{ij}$, all the ``compatible" diagonals are outside or on the light cone of the $X_{ij}$ and all of the incompatible diagonals are in the union of (strictly) future and (strictly) past light cone of $X_{ij}$. We show this beautiful map from a combinatorial 
structure to a causal structure in the figure \ref{causaldiamond6} below.
\begin{figure}[ht]
    \centering
    \includegraphics[scale=0.35]{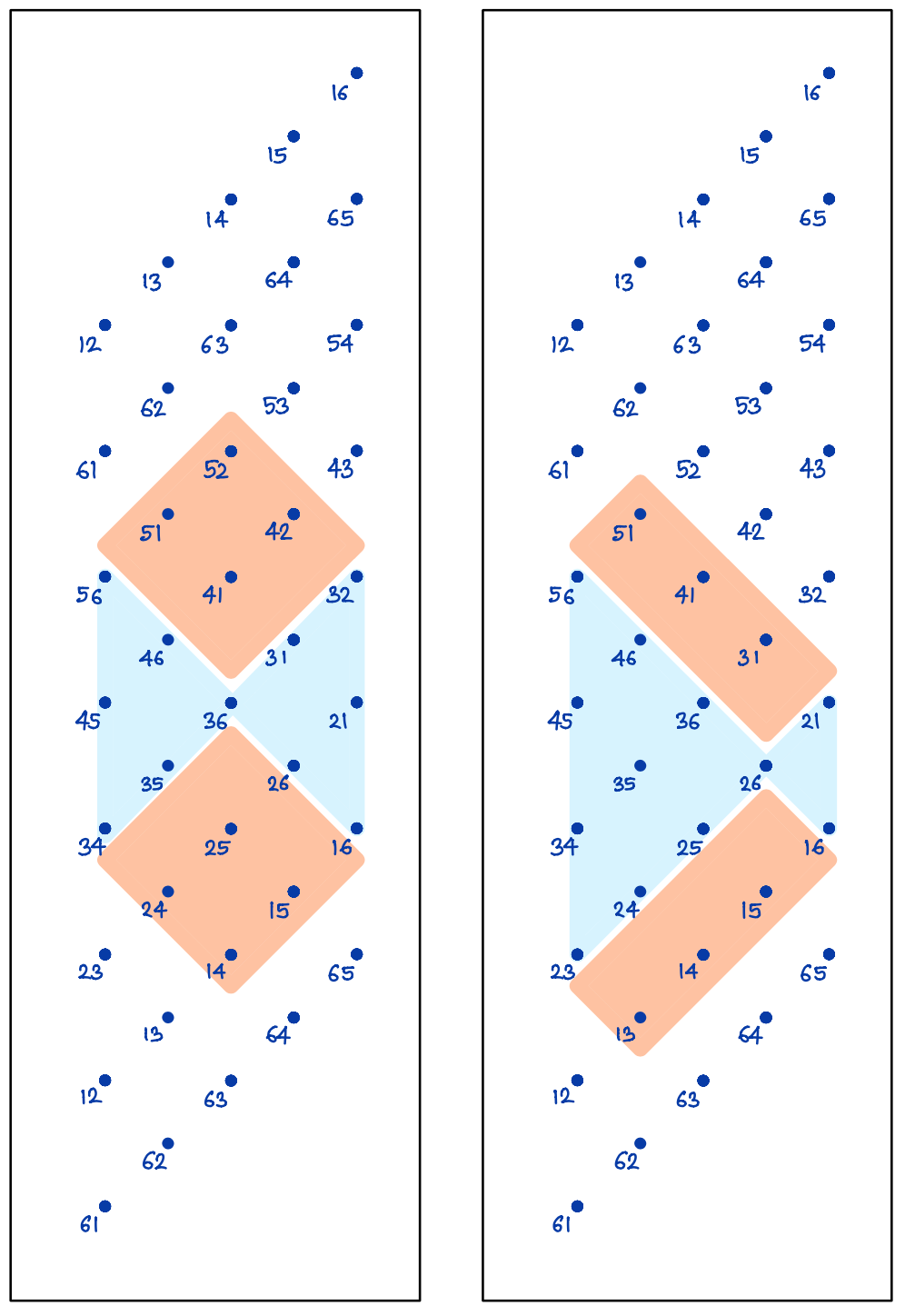}
    \caption{Compatible diagonals}
    \label{causaldiamond6}
\end{figure}

There are several features of the causal diamond which will be central to our analysis below.
\begin{itemize}
\item A temporal evolution in causal diamond corresponds to mutation such that $X_{i,j}$ and $X_{i\, \pm\, 1, j\, \pm\, 1}$ are incompatible diagonals. 
\item All the variables on a constant $u$ or $v$ slice triangulate the $n$-gon. 
\item Even after the redundancy encoded in cyclic symmetry is broken by working in a given strip bounded by $X_{i, i\pm 1}\, =\, 0$, there is a further redundancy in the causal diamond description. This is due to the fact that temporal evolution of $X_{ij}$ eventually leads us to $X_{ij}$. Thus the smallest domain in the causal diamond in which all the $X_{ij}$ variables occur precisely once is called non-redundant domain. Non-redundant domains are not unique as the initial slice (or equivalently, an initial choice of triangulation) is not unique. It is precisely the multiplicity of non-redundant domains  which lead to distinct ABHY realisations discovered and analysed in \cite{thomas-matte, pppp}. 
\item The discretized massless scalar field equation with a source precisely generate all the linear equations which leads to  convex realisation of $A_{n-3}$ in ${\cal K}_{n}^{+}$. 
\end{itemize}


\subsection{From Uncoloured to Coloured Causal Diamonds}\label{fuccd}
The causal diamonds (reviewed in section \ref{brcd}) give a novel perspective on ABHY realisations of the associahedron. There are many non-redundant domains inside the causal diamond. The non-redundant domains are the collection of vertices inside the causal diamond which is isomorphic to the set of  all the diagonals. ``Initial slice" in the non-redundant domain corresponds to a complete triangulation $T$ that generates the ABHY realisations via the constraints $s_{ij}\, =\, -\, c_{ij}\, \forall\, (i,j)\, \notin\, T^{c}$.

For the case of mostly black  triangulations  the simple minded mutation rules face obstructions. Hence we can ask if there are a different set of mutation rules which generate closed combinatorial poytopes. Rather remarkably, if we let the abstract structure of a causal diamond guide us, then it leads us to a set of mutation rules and a rather natural  definition of a (set of) positive geometries with multi-coloured facet structure. In order to adapt the causal diamond to mostly black triangulations, we let the vertices of causal diamond have one of two possible colours, namely black and red. As any mostly black  triangulation has all but one diagonals black, we start with any slice in the causal diamond which corresponds to a complete triangulation  and colour one of the vertices in it as red (see figure \ref{coloredcausalinislice}).
\begin{figure}
    \centering
    \includegraphics[scale=0.5]{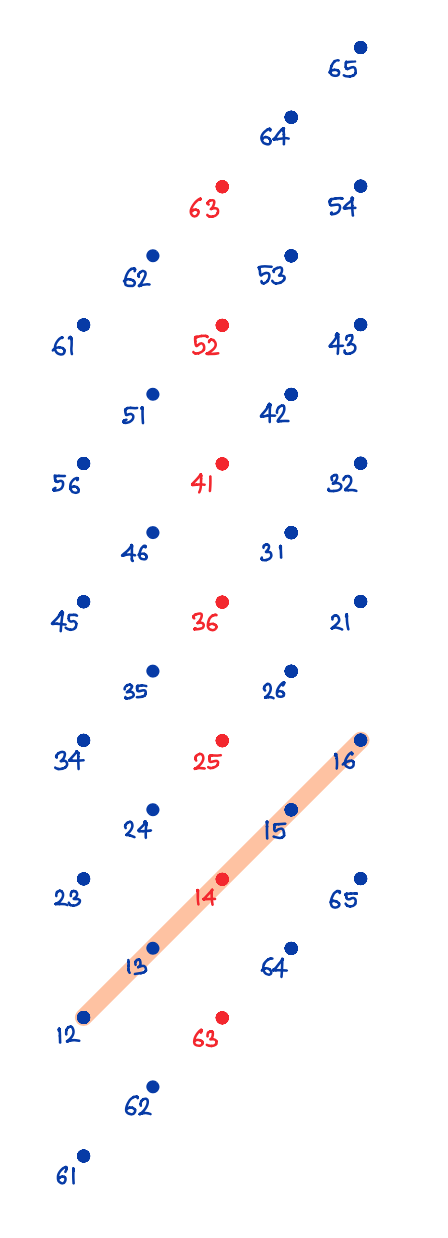}
    \caption{Colored Causal Diamond with initial slice $(13),(14),(15)$ where $(14)$ is red}
    \label{coloredcausalinislice}
\end{figure}
We then keep the mutation rules as they are for uncoloured causal diamonds so that a vertex evolves into another vertex with the same colour. \emph{We refer to the resulting diamonds as coloured causal diamonds}. Certain properties of coloured causal diamonds can immediately be noted.
\begin{itemize}
\item If any slice on the causal diamond which corresponds to a triangulation of the $n$-gon has a red vertex $X_{ij}$ then any vertex in the causal diamond will be red if $\vert k - l\vert\, =\, \vert i - j\vert\, \textrm{modulo}\, n$ and if it is in the same vertical line as $X_{ij}$. 
\item There are $\lceil \frac{n-3}{2} \rceil$ disconnected copies of coloured causal diamonds that contain all possible mostly black  triangulations. 
\end{itemize}
We now consider specific non-redundant domains inside the coloured causal diamond (described in section \ref{ccdcbk}). We refer to these domains as fundamental domains. The initial ``null slice" of the fundamental domain  has one red and all other black vertices. The final slice obtained within the fundamental domain ensures that combinatorially we have an associahedron polytope but whose facets have additional (red or black) labelling.

Remarkably enough, distinct fundamental domains (differentiated by the choice of colouring of vertices on the initial  slice) produce $\frac{n(n-3)}{2}$ number of combinatorial associahedra. Union of vertices belonging to all such associahedra contain (1) (multiple copies) of all the mono-chromatic and mostly black  triangulations and (2) do not contain any other species of triangulations (such as a triangulation with two diagonals being red). \emph{We refer to this (coloured) associahedron as associahedron block}.

As each associahedron block is identified with a fundamental domain in the coloured causal diamond, we have a realisation of each block in the (positive region) of kinematic space. One of our primary results in this paper is to show that the associahedron blocks are the positive geometries for ${\cal M}_{n}^{Y}$.
\subsection{ Coloured Causal diamonds and Convex Blocks in kinematic space}\label{ccdcbk}
We will now give a detailed analysis of the ideas outlined in section \ref{fuccd} and colour the causal diamonds in the kinematic space. We will then show that the coloured causal structure can be used to locate a class of polytopes in ${\cal K}_{n}^{+}$ which constitute the ``amplituhedron" for the amplitudes of two-scalar field theory.

We start with an ``intial configuration" where precisely one of the vertex is red. We then ``evolve the vertices'' following our mutation rule which does not change the colour.  After finite number of walks (or evolution in discrete time in the causal diamond) we reach the initial configuration again. If the initial configuration has a red vertex $X_{ij}$, then all the vertices $X_{kl}$ where $\vert\, l - k\, \vert\, =\, \vert\, j - i\, \vert$ modulo $n$ and where $X_{kl}$ is in the same vertical line as $X_{ij}$ will eventually be coloured red in this causal diamond. However, all $X_{kl}$ where $\vert\, l - k\, \vert\, \neq\, \vert j - i \vert\, \textrm{modulo}\, n$ will remain black. Hence there exists disconnected copies of coloured causal diamonds, on each of which the reference triangulation (or a lattice points on a null slice)  has precisely one red vertex $X_{ij}$. Clearly, the number of disconnected copies of coloured causal diamonds is   $\lceil\frac{n-3}{2}\rceil$.

We then consider  specific  ``non-redundant domains" (\cite{nima1912})  inside the coloured causal diamond each of which  contain all the vertices precisely once and the smallest possible number of  red vertices and each vertex $(k,l)$ (irrespective of it's colour) occur exactly once. Hence each of these domains is classified by the set of red vertices $\tilde{X}_{ij}, \tilde{X}_{i+1,j+1},\, \dots,\, \tilde{X}_{i+k, j + k}$ (where $k\, =\, \vert j - i\vert$ modulo $n$ and where we identify $n+1$ with $1$.) We refer to this class of non-redundant domains as \emph{fundamental domains} and denote them by ${\cal F}_{ij}$. Before giving some examples of the fundamental domains, we make a couple of observations :\\
{\bf (1)} Strictly speaking, the fundamental domain even with same set of red vertices are not unique and depend on initial choice of triangulation. However we will suppress the explicit dependence on the reference (initial) triangulation as  fundamental domains with same set of red vertices but different initial triangulations are simply different ABHY realisations of the same associahedron block.\\
{\bf (2)} We will always choose the initial triangulation containing red diagonal $(i,j)$ to be such that all the diagonals emanate in the same vertex $j$. Following examples of ${\cal F}_{ij}$ illustrate these ideas explicitly.

 One choice of fundamental domain in the case of $n\, =\, 5$ and $n\, =\, 6$ are given below and shown in figures \ref{fundomain5} and \ref{fundomain6}.
 
 \begin{figure}
     \centering
     \includegraphics[scale=0.4]{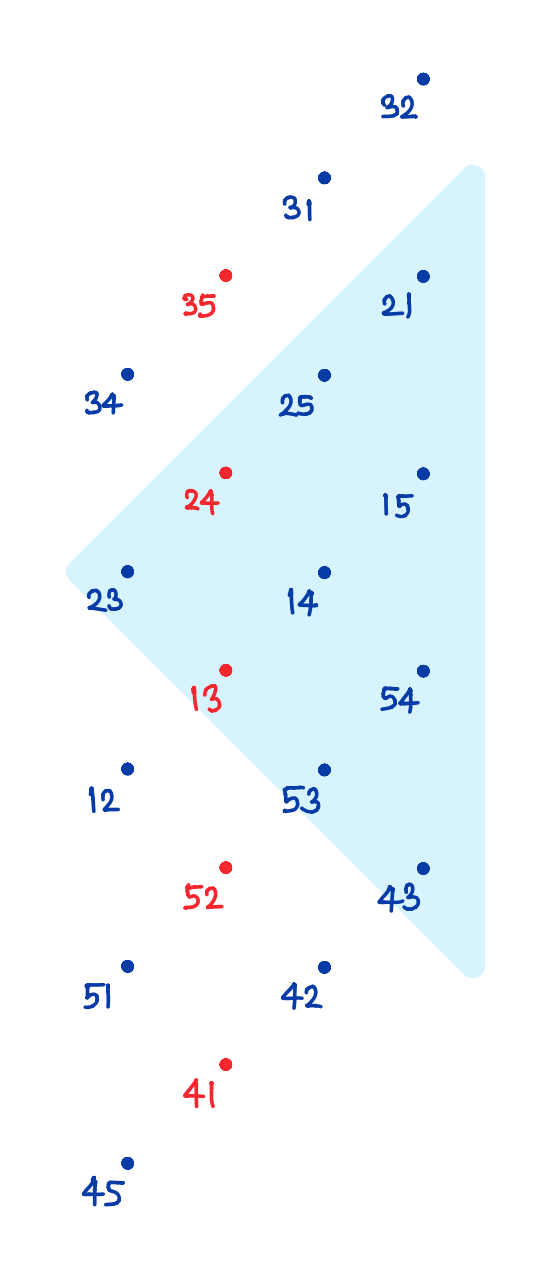}
     \caption{Fundamental domain $\mathcal{F}_{13}$ in $n=5$}
     \label{fundomain5}
 \end{figure}
 
 \begin{figure}
     \centering
     \includegraphics[scale=0.4]{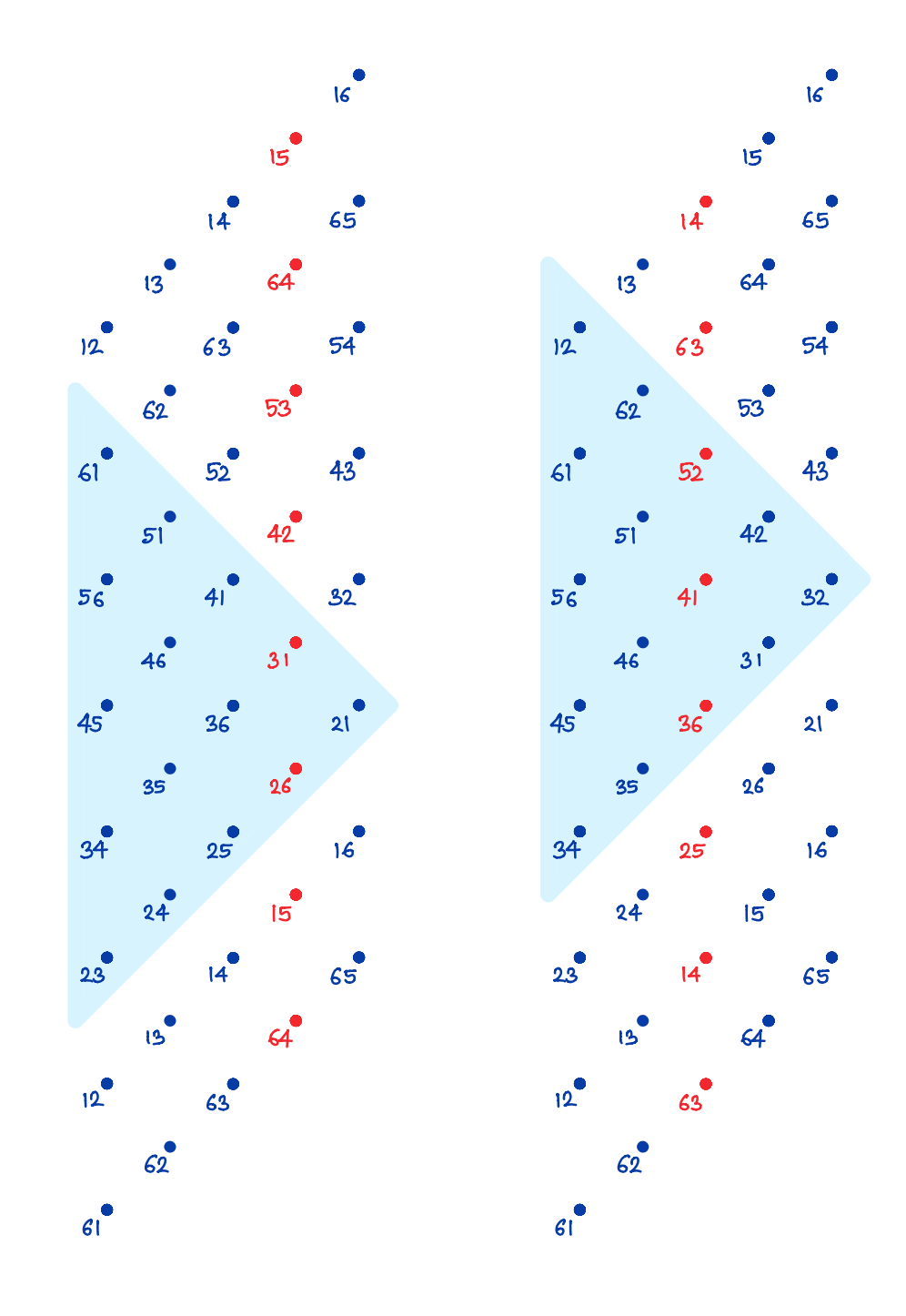}
     \caption{Fundamental domains $\mathcal{F}_{26}$ and $\mathcal{F}_{36}$ in $n=6$}
     \label{fundomain6}
 \end{figure}
 
 \begin{flalign}\label{mar9-1}
 \begin{array}{lll}
 \textrm{n = 5 fundamental domains}\, =\, \{\, {\cal F}_{13}, {\cal F}_{24}, {\cal F}_{35}, {\cal F}_{14}, {\cal F}_{25}\, \}\\
 \textrm{n=6 fundamental domains}\, =\, \{\, {\cal F}_{13}, {\cal F}_{24}, {\cal F}_{35}, {\cal F}_{46}, {\cal F}_{15}, {\cal F}_{26}\, \}\, \cup\, \{\, {\cal F}_{14} {\cal F}_{25},\, {\cal F}_{36}\, \}
 \end{array}
 \end{flalign}
 
 Several comments are in order.
 \begin{itemize}
 \item We note that  ${\cal F}_{ij}$ is a domain which contains all the diagonals of the $n$-gon such that the minimal number of them (in the case of $n\, =\, 5$, 2 vertices) are coloured red. 
 \item In the $n\, =\, 6$ case there are two disconnected copies of the coloured causal diamonds as evolution of a red vertex $(k,l)$ can never produce $(i,j)$ if $\vert j - i\vert\, \neq\, \vert k - l\vert\, \textrm{mod}\, 6$.
 \item The fundamental domains ${\cal F}_{25},\, {\cal F}_{36}$ are redundant, as the configuration of red and black diagonals in these domains is same as that in ${\cal F}_{14}$. We do not have to include them, but as we will see  this redundant inclusion  of ${\cal F}_{ij}$ for $\vert j - i\vert \, =\, \floor*{\frac{n}{2}}$ is useful for writing the final formula for (weighted) sum over canonical forms that produce the scattering amplitude. 
 \item As can be seen from the figure \ref{fundomain5} and figure \ref{fundomain6}, the examples of fundamental domains above are such that the initial and final triangulations emanate from a given vertex. But we can start with different initial triangulation. In the $n\, =\, 6$ case for example, we can consider the fundamental domain ${\cal F}_{36}$ in which the intial quiver(or the reference triangulation) is $T\, =\, \{\, 13, \textcolor{red}{36}, 35\, \}$. The corresponding Fundamental domain has precisely three red vertices $\{ \textcolor{red}{36},\, \textcolor{red}{14},\, \textcolor{red}{25}\, \}$ again and simply gives us a different realisation of the associahedron block $A_{2}^{(14)}$. 
 \end{itemize}
Let us look at the fundamental domains in $n=7$ case. We start with $\{\, (i,j),\, \dots,\, (j-1, i-1)\, \}$ where $\vert j - i \vert\, \geq\, 2$. Hence in $n=7$ case we have the following \emph{set} of \emph{disjoint} fundamental domains each of which defines a coloured causal diamond.\\
\begin{flalign}
\begin{array}{lll}
\textrm{A complete set of fundamental domains with ($\vert j - i \vert$ = 2 modulo 7})=\\
 \{\, {\cal F}_{13},\, {\cal F}_{24}, {\cal F}_{35},\, {\cal F}_{46},  {\cal F}_{57},\, {\cal F}_{16}, {\cal F}_{27}\, \}\nonumber\\
\textrm{A complete set of fundamental domains with ($\vert j - i\vert$ = 3 modulo 7)}=\\
 \{\, {\cal F}_{14},  {\cal F}_{36}, {\cal F}_{15}, {\cal F}_{37}, {\cal F}_{25}, {\cal F}_{47}, {\cal F}_{26}\,\}
\end{array}
\end{flalign}
We thus get, seven associahedron blocks ${\cal A}_{4}^{(i,j)}\vert i-j\vert\, =\, 2\, \textrm{modulo}\, 7$ with two non-adjacent red facets and seven more associahedron blocks ${\cal A}_{4}^{(m,n)}\vert m-n\vert\, =\, 3\, \textrm{modulo}\, 7$ each of which has three non-adjacent red facets.\\
\subsection{Mutation rules from Coloured Causal Diamonds}\label{mrccd}
As we have defined the associahedron block $A_{n-3}^{(i,j)}$ directly through coloured causal diamonds, their convex realisation in ${\cal K}_{n}^{+}$  is precisely given by the ABHY equations (with certain $X_{ij}$ replaced by $\tilde{X}_{ij}$ such that $\tilde{X}_{ij}\, \geq\, 0$).However, the mutation rules and the corresponding abstract definition of the combinatorial polytopes remained implicit in our construction. In this section we give the combinatorial definition of the associahedron blocks and then describe their realisations in kinematic space.

An associahedron block is an associahedron defined by a reference triangulation $T_{i,j}$ which has precisely one red diagonal $(i,j)$.  We now consider all possible triangulations obtained from $T_{i,j}$ via mutation such that ``walks along the causal diamond" (as reviewed in section \ref{ccdcp}) always produce triangulations which are mostly black . That is, they correspond to mutations which do not change the colour of the dianonal being mutated. 
However any mutation which maps a diagonal $(k,l)\, \rightarrow\, (i,j)$ is colour changing from 
\begin{flalign}
\begin{array}{lll}
{\bf (i)}\, \textrm{Red to black if $(k,l)$ labels the red vertex and $(i,j)$ does not and}\\
{\bf (ii)}\, \textrm{Black to red if $(k,l)$ labels the black vertex and $(i,j)$ is red.}\nonumber
\end{array}
\end{flalign}
We thus note that the (combinatorial) associahedron block depends on the choice of reference $(i,j)$. 

With these observations we define the associahedron block as follows :\\
An associahedron block associated to a reference diagonal $(ij)$ is combinatorially an associahedron in which the diagonals from the set $\{\, (i,j),\, \dots,\, (i + \vert j - i \vert - 1, j + \vert j - i\vert -1)\, \}$ are red while others are black. One example of the associahedron block in two dimensions is shown in the figure \ref{Assoblock514}.
\begin{figure}[h]
    \centering
    \includegraphics[scale=0.3]{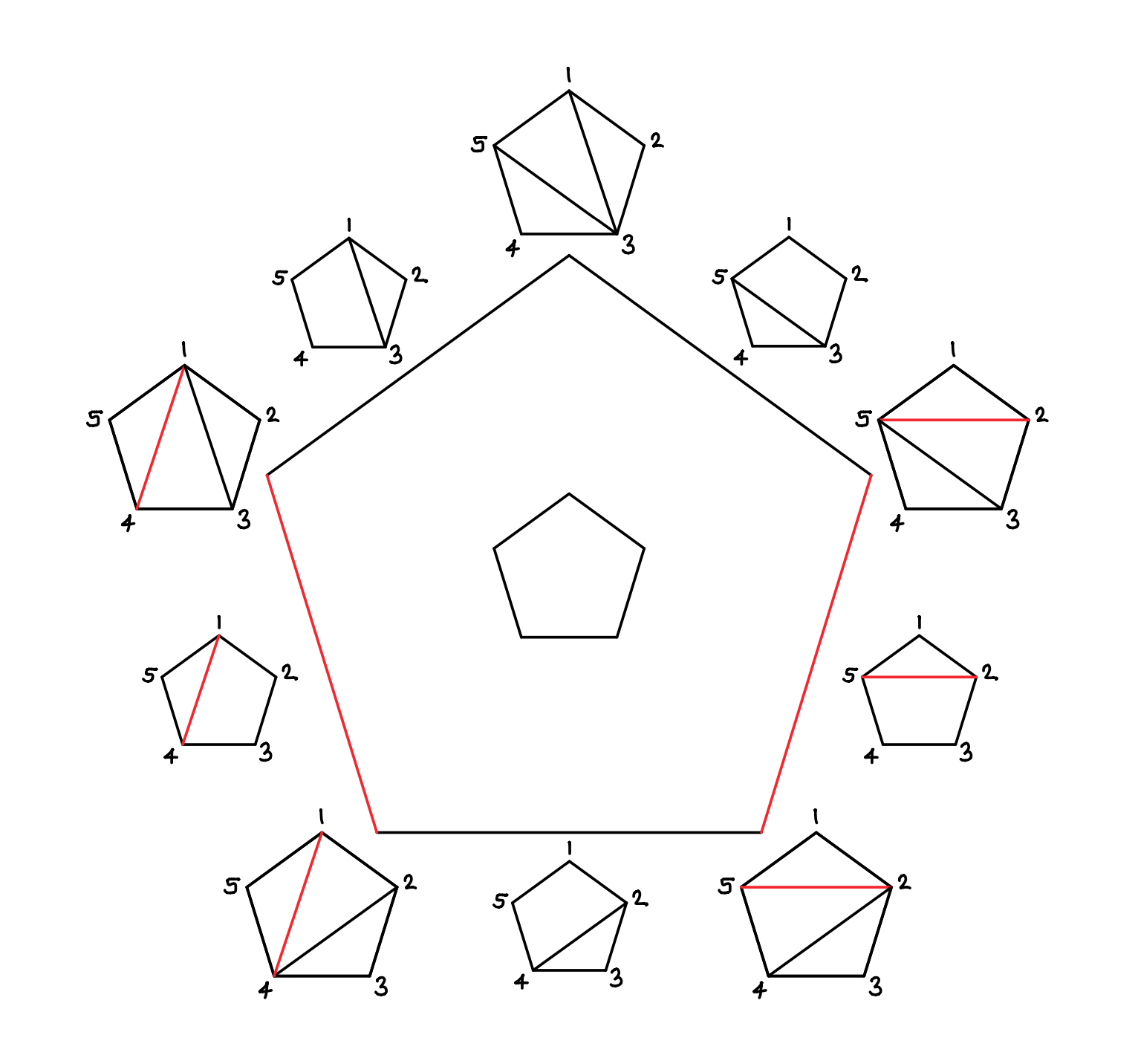}
    \caption{Associahedron block $A^{14}_{2}$}
    \label{Assoblock514}
\end{figure}
By inspection, any associahedron-block is a simple polytope that has the following properties. 
\begin{itemize}
\item Combinatorially, $A_{n-3}^{{\cal F}_{ij}}$ is an associahedron.
\item  Vertices of $A_{n-3}^{{\cal F}_{ij}}$ are in bijection with either mostly black  or uncoloured triangulation of an $n$-gon.
\item Co-dimension k faces of $A_{n-3}^{{\cal F}_{ij}}$ are in one to one correspondence with k-partial triangulation which is either mostly black  or uncoloured.
\item Each of its face is a product of a lower associahedron block and an associahedra.
\end{itemize}

We also note that the dynamical equations (ABHY constraints) governing the evolution inside the fundamental domain generates convex realisations of each of the blocks. In particular for $ n = 5$, the evolution equations inside the fundamental domain ${\cal F}_{13}$  generates the following realisation of an associahedron in kinematic space.
\begin{flalign}
\begin{array}{lll}
X_{13}\, +\, X_{24}\, -\, \tilde{X}_{14}\, =\, c_{13}\\
\tilde{X}_{14}\, +\, X_{35}\, =\, c_{14} +c_{24}\\
X_{13}\, +\, \tilde{X}_{25}\, =\, c_{13}+c_{14}
\end{array}
\end{flalign}
These equations admit solutions in the positive region ${\cal K}_{n}^{+}$ as long as $c_{ij}\, >\, 0$. 

In general, given a fundamental domain ${\cal F}_{ij}$ in which only $\{\, (i,j),\, \dots,\, (i + \vert j - i \vert - 1, j + \vert j - i\vert -1)\, \}$ are red,  the reference triangulation generated by initial configuration consists of $T\, =\, \{\, (j-2,j),\, \dots,\, (i,j),\, \dots,\, (j+2, j)\, \}$ (where without loss of generality we have assumed that if  $j - i\, =\, k$ mod $n$ then $k\, <\, n-k$.) Hence the ABHY equations which produce the convex realisation of the corresponding associahedron block can be written as follows. We  first define abstract variables ${\bf s}_{ij}$ as,
\begin{flalign}
{\bf s}_{ij}\, =\, {\bf X}_{ij}\, +\, {\bf X}_{i+1,j+1}\, -\, {\bf X}_{i, j+1}\, -\, {\bf X}_{i+1,j}
\end{flalign}
where ${\bf X}_{mn}\, =\, \tilde{X}_{mn},\, \textrm{if} (m,n)\, \in\, S_{(i,j)}$ and ${\bf X}_{mn}\, =\, X_{mn}$ otherwise. The convex realisation of the associahedron block is generated by 
\begin{flalign}
{\bf s}_{ij}\, =\, -\, c_{ij}\, \forall\, (i,j)\, \notin\, T^{c}
\end{flalign}
where $c_{ij}$ are positive constants such that $X_{mn}\, \geq\, m^{2}$ for $(m,n)\, \in\, S_{i,j}$ and $X_{mn}\, \geq\, 0$ otherwise. 

\subsection{From the Associahedron Blocks to Scattering Amplitudes}\label{absa}
In this section, we prove one of our primary results. We show that the $n$-point tree-level (and colour-ordered) scattering amplitudes in our theory when expanded to sub-leading order in $\lambda_{2}$ are a \emph{unique} weighted sum over canonical forms of all the associahedron blocks in a given dimension
As each block is an (ABHY) associahedron, there is a canonical form associated to it and a weighted sum of all the canonical forms when added to  the form associated to ABHY associahedron for $\phi_{1}^{3}$ theory produce the desired amplitude. We start with the example of two dimensional associahedron-blocks (i.e. $n\, =\, 5$ case) before proving the master formula in equations (\ref{master0}, \ref{master1}). 

In the $n\, =\, 5$ example, we have five associahedron blocks
\begin{flalign}
\{\, A_{2}^{(i,j)}\vert\, (i,j)\, \in\, \{ (1,3), (2,4), (3,5), (1,4), (2.5)\, \}\nonumber
\end{flalign}
each of which has two non-adjacent red facets that correspond to $\tilde{X}_{ij}\, =\, \tilde{X}_{i+1,j+1}\, =\, 0$. There is a  unique canonical form on ${\cal K}_{2}$ defined by each of these associahedron blocks, which are simple polytopes. \cite{nima1703}. As an example, we consider the planar scattering form defined by $A_{2}^{(1,3)}$. We recall that this is an ABHY associahedron two of whose facets correspond to the massive poles $X_{13}\, =\, X_{24}\, =\,m^{2}$.
\begin{flalign}
\begin{array}{lll}
\Omega_{n=5}^{A_{2}^{{\cal F}_{13}}}\, =\\
d\log\tilde{X}_{13}\, \wedge\, d\log X_{14}\, -\, d\log\tilde{X}_{13}\, \wedge\, d\log X_{35}\, +\, d\log X_{25}\, \wedge\, d\log X_{35}\, -\\ 
\hspace*{2.3in} \, d\log X_{25}\, \wedge\, d\log\tilde{X}_{24}\, +\, d\log X_{14}\, \wedge\, d\log\tilde{X}_{24}
\end{array}
\end{flalign}
The forms associated to other blocks can be written down similarly. It is  now easy to verify that the following sum of canonical forms restricted to the five associahedron blocks, 
\begin{flalign}
\Omega_{n=5}^{Y}\, :=\, \frac{1}{2}\, \sum_{(i,j)}^{5}\, \Omega_{n=5}^{(i,j)}\vert_{A_{2}^{(i,j)}} -\, \frac{1}{2}\, \Omega_{n=5}\vert_{A_{2}}
\end{flalign}
generate the contribution to the scattering amplitude at $\lambda_{1}\lambda_{2}^{2}$ order. That is, it produces all the terms which has one massless and one massive pole. And hence the following form generates the complete amplitude upto $O(\lambda_{2}^{2})$.
\begin{flalign}
\omega_{n=5}\, :=\, \lambda_{1}\lambda_{2}^{2}\, \Omega_{n=5}^{Y}\, +\, \lambda_{1}^{3}\, \Omega_{n=5}\vert_{{A}_{2}}
\end{flalign}
It is a happy surprise to see that there exists a linear combination of the canonical forms on the associahedron blocks  $\Omega_{n=5}^{Y}$ that produces the amplitude contribution \emph{at} $\lambda_{1}\lambda_{2}^{2}$ order. That is, it contains only those channels which have one massive and one massless pole. The existence of $\Omega_{n=5}^{Y}$ is thus necessary and sufficient to obtain the tree-level amplitude for independent couplings $\lambda_{1},\, \lambda_{2}$. However, as the number of particles increase, so do the number of assocaihedron blocks and then it is not at all obvious if there exists any formula for $\Omega_{n\, >\, 5}^{Y}$. We now derive a general formula for any $n$ and verify it using a few non-trivial examples.
\begin{lemma}
The following weighted sum of canonical forms restricted to their respective associahedron blocks generate the contribution to the amplitude at $\lambda_{1}^{n-4}\, \lambda_{2}^{2}$.  
\begin{flalign}
\Omega_{n}^{Y}\, :=\, \left[\, \sum_{\vert i - j\vert=2}^{\floor*{\frac{n}{2}}}\, \frac{1}{j-i}\, \sum_{{\cal F}_{ij}}\, \Omega_{n}^{(i,j)}(A_{n-3}^{{\cal F}_{ij}}) -\, \gamma\, \Omega_{n}^{\phi^{3}}(A_{n-3}) \right]
\end{flalign}
where 
\begin{flalign}
\gamma\, =\, \sum_{\textrm{sum over all diagonals, (i,j)}}\, \frac{1}{\vert j - i\vert}\, -\, (n-3)
\end{flalign}
\end{lemma}
\begin{proof}
When we sum over fundamental domains ${\cal F}_{ij}$ and span over all the associahedron-blocks, each vertex (complete triangulation) is covered at least once. These includes monochromatic triangulations where none of the diagonals are red. Hence finding an $\Omega_{n}^{Y}$ amounts to showing that there is a \emph{unique} choice of co-efficients which removes the contribution of all the mono chromatic triangulations while simultaneously ensuring that each mostly black  triangulation contributes precisely once. We determine these co-efficients as follows.\\
In the fundamental domain, ${\cal F}_{ij}$, all the vertices co-ordinatized by $I\, =\, \{\, (i,j),\,(i+1, j+1)\, \dots,\, (j-1, 2j\, -\, i\, -1)\, \}$ are red and the remaining vertices are black. Hence  given ${\cal F}_{ij}$, any triangulation which has one of the diagonals in $I$ is a mostly black  triangulation and all the other triangulations (which contain no elements from $I$) are monochromatic triangulations. Now as we scan over all of the fundamental domains, the vertex $(i,j)$ is labelled red in $\vert\, j\, -\, i\, \vert$ number of fundamental domains. Hence any mostly black  triangulation ( i.e. vertex of the associahedron  block) containing  $\tilde{X}_{ij}$ occurs $\vert\, j - i\, \vert$ times. We thus see that the \emph{unique} weighted sum over the canonical forms $\Omega_{n}^{{\cal F}_{ij}}\vert_{{\cal A}_{n-3}^{ij}}$  which generates all the terms in the amplitude that have one massive and $n-4$ massless poles is given by,
\begin{flalign}
\sum_{\vert i - j\vert\, =\, 2}^{\floor*{\frac{n}{2}}}\, \frac{1}{\vert\, i - j\, \vert}\sum_{{\cal F}_{ij}\, \vert \vert i - j \vert\, \textrm{modulo}\, n}\, \Omega_{n}^{{\cal F}_{ij}}
\end{flalign}
The uniqueness of the weights simply follows by inspection. As all the $(i,j)$ for which  $\vert j - i\vert$ modulo $n$ is the same occur with frequency $\vert j - i\vert$, there is no other choice of weights for which we get unit residue over all channels.

However the above sum will contain contribution over mono-chromatic triangulations which are in bijection with purely massless channels with varying weights. We now argue that all the mono-chromatic triangulations in fact occur with equal frequency $\gamma$ and hence can be removed by simply subtracting $\gamma\, \Omega_{n}({A_{n}})$ from the above form. This then defines $\Omega_{n}^{Y}$ and completes our proof.

Let us for a moment assume that there exists such a $\Omega_{n}^{Y}$ . Then 
\begin{flalign}\label{ap5-0}
\sum_{\vert, i - j\, \vert = 2}^{\frac{n}{2}}\, \frac{1}{\vert, i - j\, \vert}\sum_{{\cal F}_{ij}\vert (j-i)\, \textrm{modulo}\, n}\, \Omega_{n}^{{\cal F}_{ij}}({\cal A}_{n-3}^{ij})\, -\, \Omega_{n}^{Y}
\end{flalign}
contains is the form of ABHY associahedron (in which all the vertices correspond to massless poles with residue one). 
Let us now remove all the red-colouring from the mostly black triangulations turning each such mostly black  triangulation into a vertex of the (uncoloured) associahedron. This would imply that 
\begin{flalign}\label{ap5-1}
\Omega_{n}^{{\cal F}_{ij}}\, \rightarrow\, \Omega_{n}\, \forall\, (i,j)\\
\Omega_{n}^{Y}\, \rightarrow\, (n-3)\, \Omega_{n}
\end{flalign}
where the second equation follows from the fact that there are $n-3$ mostly black  triangulations that map to monochromatic triangulation when we strip off the colour. As the map of mostly black  to monochromatic triangulation does not effect the massless channels, we have the following result. 
\begin{flalign}
\sum_{\vert i - j\vert\, =1}^{\floor*{\frac{n}{2}}}\, \frac{1}{\vert i - j\, \vert}\, \sum_{{\cal F}_{ij}}\, \Omega_{n}(A_{n})\, -\, (n-3)\, \Omega_{n}(A_{n})\, =\, \gamma\, \Omega_{n}
\end{flalign}
and hence $\gamma$ is fixed to 
\begin{flalign}
\gamma\, =\, \sum_{\vert\, i - j\, \vert =1}^{\floor*{\frac{n}{2}}}\, \frac{1}{\, \vert i - j\, \vert}\, \sum_{{\cal F}_{ij}}\, 1 - (n-3)
\end{flalign}

\end{proof}
\vspace*{-0.3in}
We consider  the equation \eqref{ap5-0} to be a striking result. It is apriori not obvious that the scattering amplitude ${\cal M}_{n}^{\textrm{CO} (2)}$ (defined in equation \eqref{march8-1}) should be a sum over canonical forms of associahedron blocks. Hence we view the result of the above lemma as a non-trivial evidence in support of the universality of the Amplituhedron program. 

One can verify the validity of the master formula in certain lower dimensional cases. Here we give two explicit examples of $n = 6$ and $n = 7$. 
\begin{flalign}
\Omega_{n=6}^{Y}\, =\, \frac{1}{2}\sum_{{\cal F}_{ij}\vert j - i\vert \textrm{even mod 6}}\, \Omega_{6}^{{\cal F}_{ij}}\, +\, \frac{1}{3}\sum_{{\cal F}_{ij}\vert j - i\vert \textrm{odd mod 6}}\, \Omega_{6}^{{\cal F}_{ij}}\, -\, \Omega_{\phi^{3}}\\
 \Omega_{n=7}^{Y}\, =\, \frac{1}{2}\sum_{{\cal F}_{ij}\vert j - i\vert \textrm{2 mod 7}}\, \Omega_{6}^{{\cal F}_{ij}}\, +\, \frac{1}{3}\sum_{{\cal F}_{ij}\vert j - i\vert \textrm{3 mod 7}}\, \Omega_{6}^{{\cal F}_{ij}}\, -\, \frac{11}{6} \,\Omega_{\phi^{3}}
\end{flalign}
It can be in fact verified that these formulae indeed produce the correct scattering amplitudes.

We thus see that the remarkable richness of such an ``doubled" causal diamond structure (where each vertex can be either black or red) lies in  the generalisation of the above construction to higher dimensional polytopes.\\

\subsection{Factorisation on Massive Channels}\label{fmc}
It was shown in \cite{abhy1711} that combinatorial factorisation of an associahedron implies the geometric factorisation of it's convex realisation in ${\cal K}_{n}^{+}$. The geometric factorisation  can be quantified as,
\begin{flalign}
A_{n-3}\, \vert_{X_{ij}\, =\, 0}\, =\, A_{L}(i, i +1,\, \dots\, \overline{I})\, \times\, A_{R}(1,\, \dots,\, i-1\, I,\, \dots,\, n)
\end{flalign}
$A_{L},\, A_{R}$ are ABHY associahedra of dimensions $\vert j - i -2\vert$ and $n - (j - i)\, -\, 1$ respectively. 
ABHY proved that the geometric factorisation of the associahedron leads to the factorisation of amplitudes through the residue formula for the canonical forms.
\begin{flalign}
\textrm{Res}_{X_{ij}\, =\, 0}\, \Omega_{n-3}(A_{n-3})\, =\, \Omega(A_{L})\, \wedge\, \Omega(A_{R})
\end{flalign}
where $\Omega(A_{L}),\, \Omega(A_{R})$ are canonical forms associated to the lower dimensional associahedra $A_{L},\, A_{R}$ respectively. These results generalise rather trivially to the associahedron blocks reaffirming one of the striking results of the amplituhedron program :  Locality and Unitarity are consequences of the underlying positive geometry.

As we saw in an earlier section, a co-dimension one boundary of an associahedron block is either a product of two lower dimensional associahedron blocks, or an associahedron block and an associahedron or product of two associahedra.\footnote{A vertex is zero dimensional associahedron and hence we identify an associahedron with cartesian product of the associahedron with a zero dimensional associahedron.} It can be verified that exactly as in the case of associahedron \cite{abhy1711} this combinatorial factorisation of the associahedron-block implies it's geometric factorisation as well.  The geometric factorisation can be explicitly written as 
\begin{flalign}
\begin{array}{lll}
A_{n-3}^{{\cal F}_{ij}}\, \vert_{X_{mn}\, =\, 0}\, &=\,A_{L}^{{\cal F}^{L}_{ij}}(m,m+1, \dots,\, \overline{I})\, \times\, A_{R}^{{\cal F}_{ij}^{R}}(1,\, \dots,\, m-1,\, I\, \dots,\, n)\, \textrm{or}\\
&=\,A_{L}^{{\cal F}_{ij}}(m,m+1, \dots,\, \overline{\textcolor{red}{I}})\, \times\, A_{R}^{{\cal F}_{ij}}(1,\, \dots,\, m-1,\, \textcolor{red}{I}\, \dots,\, n)
\end{array}
\end{flalign}
In the first case, we are on a boundary associated to massless pole, $X_{ij}\, =\, 0$ and $A_{L/R}^{{\cal F}_{ij}^{L/R}}$ corresponds to the lower dimensional associahedron blocks. ${\cal F}_{ij}^{L/R}$ indicates the set of red diagonals which are to the left (right) of the diagonal $(m,n)$. If ${\cal F}_{ij}^{L/R}$ is empty then the corresponding associahedron block is simply a convex ABHY associahedron. In the second case, we consider the boundary facet associated to a massive pole. In this case, both the factorized positive geometries are associahedra whose canonical forms are the scattering amplitudes involving one massive and remaining massless particles. The fact that we are on $X_{mn}\, =\, m^{2}$ hyper-plane automatically implies that the intermediate particle $I$ is massive.

From the geometric factorisation of the associahedron block, we get the factorisation property of the amplitudes.
\begin{flalign}
\begin{array}{lll}
\textrm{Res}_{X_{mn}\, =\, 0}\, \Omega_{n-3}(A_{n-3}^{{\cal F}_{ij}}) =\, \Omega_{L}(A_{L}^{{\cal F}_{ij}^{L}}\, \wedge\, \Omega_{R}(A_{R}^{{\cal F}_{ij}^{R}})\nonumber\\
\textrm{Res}_{X_{mn}\, =\, m^{2}}\, \Omega_{n-3}(A_{n-3}^{{\cal F}_{ij}}) =\, \Omega_{L}(A_{L})\, \wedge\, \Omega_{R}(A_{R})
\end{array}
\end{flalign}
The residue of the $n-3$ canonical form on black or red facets equals product of lower forms associated to associahedron blocks or associahedra respectively.\\
We thus see that just as in the case of massless scattering amplitudes the locality and unitarity of the S-matrix (even in the presence of intermediate massive states), follows from the combinatorial factorisation property of the associahedron!

\section{The Scattering form for an Effective Field Theory : Accordiohedra in low energy limit}\label{sfeft}
In this section we study the $m\, \rightarrow\, \infty$ limit of the associahedron blocks. We claim that every block projects onto a  \emph{set of} simple polytopes in this limit as the block ``moves towards infinity in various possible directions". This gives us a natural perspective on emergence of positive geometries  from associahedra in the effective field theory when massive field is integrated out! The ABHY realisations basically guarantees that a specific direction dependent $m\, \rightarrow\, \infty$ limit is a bijection from the associahedron block onto an the polytope that we will refer to as ``projected accordiohedron".

Let us first consider the fundamental domain ${\cal F}_{13}$ whose initial quiver corresponds to  the mostly black  triangulation $\{\, \textcolor{red}{13}, 35\, \}$. The ABHY realisation of the corresponding associahedron block is given by,
\begin{flalign}\label{abhyn513}
\tilde{X}_{13}\, +\, \tilde{X}_{24}\, -\, X_{14}\, =\, c_{13}\\
X_{14}\, +\, X_{25}\, -\, \tilde{X}_{24}\, =\, c_{14}\\
X_{35}\, +\, X_{14}\, -\, \tilde{X}_{13}\, =\, c_{35}
\end{flalign}
where all $c_{ij}\, >\, 0$.

For the purpose of this section, we will\\
{\bf (i)} Write the ABHY constraints in terms $X_{ij}$ instead of $\tilde{X}_{ij}\, =\, X_{ij}\, -,m^{2}$ and\\
{\bf (ii)} Assume that all the $c_{ij}\, >\, m^{2}$. This non-trivial bound on $c_{ij}$ ensures that we can scan the hyper-planes satisfying ABHY constraints where one of the $m^{2}\, >X_{ij}\, >\, 0$.

Equation \eqref{abhyn513} can be written as,
\begin{flalign}
X_{13}\, +\, X_{24}\, -\, X_{14}\, =\, c_{13} + 2 m^{2}\\
X_{14}\, +\, X_{25}\, -\, X_{24}\, =\, c_{14} - m^{2}\\
X_{35}\, +\, X_{14}\, -\, X_{13}\, =\, c_{35} - m^{2}
\end{flalign}
We now consider the domain in positive region of kinematic space (but outside $A_{2}^{F_{13}}$) parametrized by $X_{13}\, <<\, m^{2}$. In this case, we immediately see that the above equations reduce to,
\begin{flalign}
X_{14}\, +\, X_{35}\, =\, c_{35}\, -\, m^{2}\\
X_{25}\, =\, c_{13} + c_{14} + m^{2}\\
X_{24}\, =\, X_{14}\, +\, c_{13}\, +\, 2 m^{2}
\end{flalign}
We thus see that for $\overline{c}_{35}\, =\, c_{35}\, -\, m^{2}$ we have a one dimensional accordiohedra given by
\begin{flalign}
X_{14}\, +\, X_{35}\, =\, \overline{c}_{35}
\end{flalign}
$X_{25}$ is a positive constant and the pole $X_{24}\, =\, 0$ is ruled out as $X_{14}\, \geq\, 0$. We thus see that viewed from ``infinite distance" in kinematic space when $)\, \leq\, X_{13}\, <<\, m^{2}$, the associahedron block $A_{2}^{{\cal F}_{13}}$ ``projects onto" a one dimensional accordiohedron with boundaries as $X_{14},\, X_{35}\, =\, 0$ respectively.

We now give another example : Consider  $A_{2}^{F_{35}}$ associated to the fundamental domain ${\cal F}_{35}$. When this block is realised in $X_{13}, X_{14}$ positive quadrant, with $\tilde{X}_{14} = 0,\, \tilde{X}_{35} = 0$ moving towards infinity in $X_{14}$ direction such that for $X_{14}\, <<\, m^{2}$ we get $X_{13}, X_{24}$ accordiohedron. This can be seen as follows. The ABHY equations are,
\begin{align}\nonumber
X_{14}+X_{35} &=\, c_{14}+c_{24}\endline
X_{13}\, +\, X_{25}\, &=\, c_{13}+c_{14}\endline
X_{13}\, +\, X_{24}\, -\, X_{14} &= c_{13} .
\end{align}
Where $c_{14}+c_{24}> 2m^{2}$, as we want the two dimensional space cut out by these equations to intersect $X_{14}=m^2$ with other $X_{ij}>0$ and $X_{35}=m^2$ with other $X_{ij} > 0$. If $X_{14}\, <<\, m^{2}$ we see from the first equation that $X_{35}\, \approx \, c_{14} + c_{24} > 2m^2 $, while the second and the third equations imply, 
\begin{align}
X_{13} + X_{24}\, &=\, c_{13},
\end{align}
and $X_{25}= X_{24} + c_{14}$. Therefore, region with $X_{14}\, <<\, m^{2}$ where $X_{35} >m^2$ and all other $X_{ij}>0$ is the one dimensional accordiohedron $(24, 13)$ given by,
\begin{align}
X_{13} + X_{24}\, &=\, c.
\end{align}
\begin{figure}[H]
    \centering
    \includegraphics[scale=0.4]{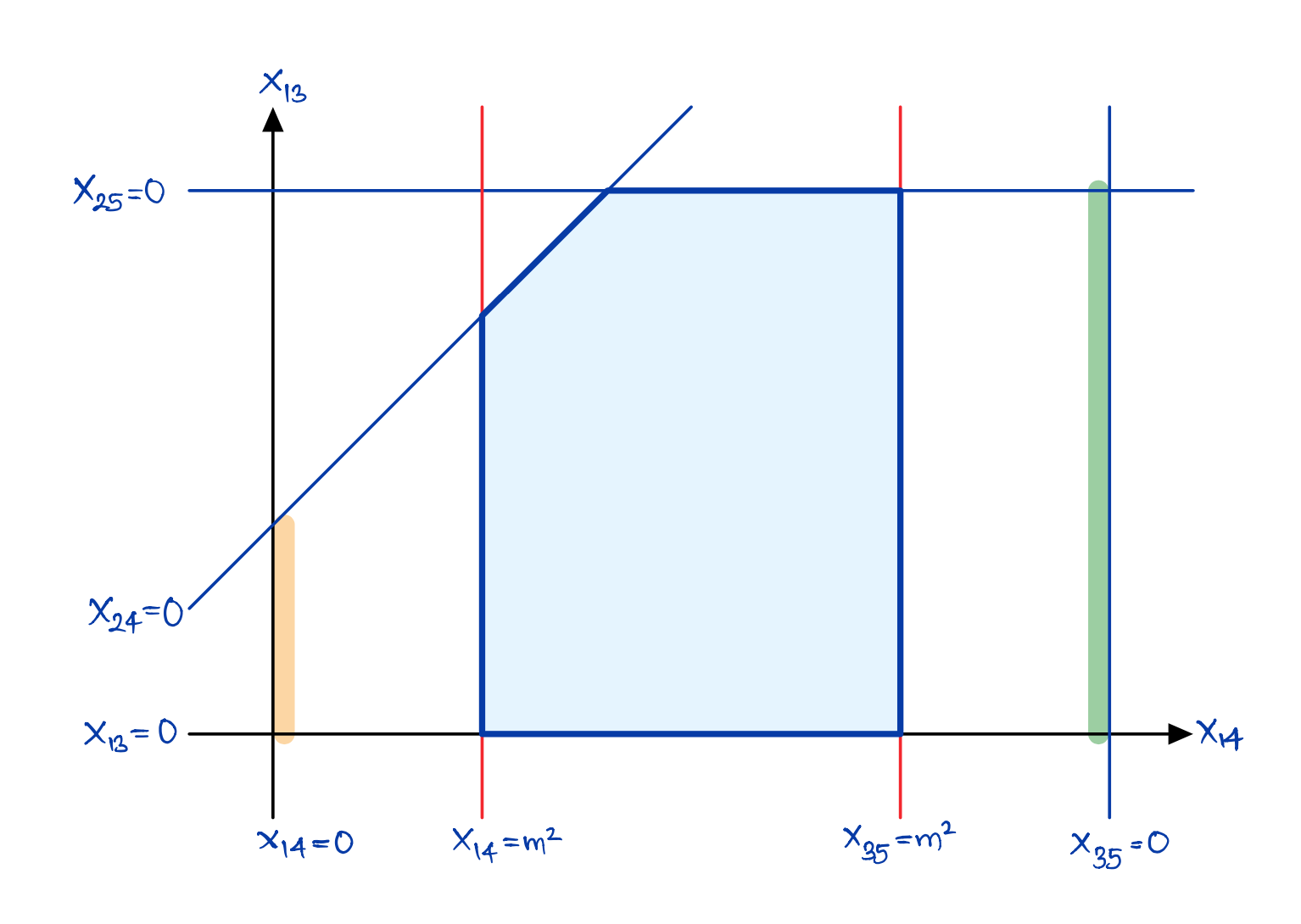}
    \caption{Acoordiohedra from Associahedra. The Orange strip is the region with $X_{14}<< m^{2}$ where we get the accordiohedron $(24,13)$, and the green strip is the region with $X_{35}<<m^{2}$ where we get the accordiohedron $(13,25)$ }
    \label{ACCasbndofAS}
\end{figure}

On the other hand, if $X_{35}\, <<\, m^{2}$ then using similar logic we see that,
\begin{flalign}
X_{13}\, +\, X_{25}\, =\, \textrm{const}
\end{flalign}
is the accordiohedron that one projects onto. This is because in this kinematic regime, we obtain 
\begin{flalign}
X_{13}\, +\, X_{25}\, =\, c_{13}\, +\, c_{14}
\end{flalign}
and $X_{24} = X_{25} + c_{24}$. Therefore, region with $X_{35}\, <<\, m^{2}$ where $X_{14} >m^2$ and all other $X_{ij}>0$ is the one dimensional accordiohedron $(13, 25)$ given by,
\begin{align}
X_{13} + X_{25}\, &=\, c.
\end{align}
As there are Five blocks in two dimensions, in the $m^{2}\, \rightarrow\, \infty$ limit, we get ten one dimensional accordiohedra. Each of the following accordiohedra occurs twice :
\begin{flalign}
\{ (13),\, (24)\, \},\, (13, 25),\, (14, 25),\, (14, 35),\, (24, 35)\, \}\nonumber
\end{flalign}

Now let's look at an example at $n=6$. Consider the three dimensional block associated to ${\cal F}_{14}$.  We consider the geometric realisation of the three dimensional block in the kinematic space given by,
\begin{align}
X_{13}+X_{26} &= c_{13}+c_{14}+c_{15} & X_{35}+X_{14}-X_{15} & = c_{14}+c_{24} \endline
X_{14}+X_{36} &= c_{14}+c_{15}+c_{24}+c_{25} & X_{25}+X_{13}-X_{15} &= c_{13}+c_{14} \endline 
X_{15} + X_{46} &= c_{15}+c_{25}+c_{35} & X_{24}+X_{13}-X_{14}&= c_{13}.
\end{align}
\begin{align}
X_{13}&>0 & X_{24}&>0 & X_{35}&>0 \endline
X_{46}&>0 & X_{15}&>0 & X_{26}&>0 \endline
X_{14}&>m^{2} & X_{25}&>m^{2} & X_{36}&>m^{2},
\end{align}
such that  $c_{14}+c_{15}+c_{24}+c_{25}  > 2 m^{2}, c_{13}+c_{14}  > m^2.$ 
If we take $X_{14}\, <<\, m^{2}$ we get the following  equations,
\begin{align}
X_{26} &\approx X_{24}+ c_{14}+c_{15} & X_{35}  & \approx X_{15} + c_{14}+c_{24} \endline
X_{36} & \approx c_{14}+c_{15}+c_{24}+c_{25} > 2m^{2} & X_{25} & \approx X_{15}+X_{24} + c_{14} > 0 \endline 
X_{15} + X_{46} &= c_{15}+c_{25}+c_{35} & X_{24}+X_{13} & \approx c_{13}.
\end{align}
 Therefore the region where $X_{14} << m^{2}$ and all other $X_{ij} > 0$  is the two dimensional accordiohedra $ (24,13) \times (15,46) $, given by 
 \begin{align}
X_{15} + X_{46} &= c_{15}+c_{25}+c_{35} & X_{24}+X_{13} & = c_{13}.
\end{align}

In general, if in a given $n-3 $ dimensional associahedron block $\tilde{X}_{ij}, \tilde{X}_{i+1, j+1},\, \tilde{X}_{i + (j - i), j + (j - i)}$ are the red vertices then taking $X_{ij}\, <<\, m^{2}$, we get a $n-4$ dimensional accordiohedra parametrized by facets which are generated by (colorless) diagonals that don't  intersect the diagonal $(ij)$.

\section{EFT Amplitude from Projected Accordiohedra}\label{eapa}
Tree-level amplitudes with mixed interactions of massless scalars have been analysed previously in the literature \cite{mrunmay3}. At each order in perturbative expansion, the positive geometry associated to the scattering amplitude generated by $\lambda \phi_{1}^{3}\, +\, g\, \phi_{1}^{4}$ interaction is an accordiohedron \cite{pilaud-1702}. In particular at order $g$ in the quartic coupling, each accordiahedron ${\cal AC}_{n-4}({\cal D})$ is a simple polytope which is defined using a reference dissection ${\cal D}$ of an $n$-gon into $n-4$ triangles and one quadrilateral. It is thus a natural question to ask if there is a bijection between  ``projected accordiohedra" ${\cal PAC}_{n-4}[(i,j)]$ that emerge in  the ``low energy"  regions of kinematic space and the accordiohedra ${\cal AC}_{n-4}({\cal D})$.

Naive expectation for such a bijection is the following :  Consider the reference triangulation $T_{ij}$ defined by the reference triangulation in ${\cal F}_{ij}$. Consider e.g. one of the projected accordiohedra, say ${\cal PAC}_{n-4}[(i,j)]$ that emerges in $X_{ij}\, <<\, m^{2}$ region.
We can interpret the low energy limit as projecting $T_{ij}$ on to a reference dissection which consists of $n-4$ triangles and one quadrilateral obtained by removing the $(i,j)$ diagonal. The dissection $T_{ij}/(i,j)$ can be used to define the ${\cal AC}_{n-4}(T_{ij}/(i,j))$ as in \cite{pilaud-1702}. We may expect the resulting polytope to be  isomorphic to the projected accordiohedron. However this expectation turns out to be wrong. 
This can immediately by checked by specific examples.
\begin{itemize}
\item Consider the $n\, =\, 5$ case where we obtained 5 associahedron blocks and 10 accordiohedra in the low energy limit. Consider one of the blocks say ${\cal A}_{2}^{{\cal F}_{13}}$ in which $\tilde{X}_{13}$ and $\tilde{X}_{24}$ are red . The initial quiver in ${\cal F}_{13}$ is $\{(1,3)_{R}, (3,5)\, \}$. Taking $X_{13}\, <<\, m^{2}$ amounts to deleting the $(1,3)$ diagonal and there by obtaining the reference dissection $(3,5)$. However, as can be readily verified using the definition of ${\cal AC}_{1}$,  the compatible set of  diagonals is $\{(3,5), (2,4)\, \}$ instead of $\{\, (1,4), (3,5)\, \}$. That is, the ${\cal AC}_{1}$ defined using $(3,5)$ as a reference is a one dimensional polytope with vertices at $X_{35}\, =\, X_{24}\, =\, 0$ as opposed to the ${\cal PAC}_{1}$ we obtained with vertices at $X_{35}\, =\, 0$ and $X_{14}\, =\, 0$.
\item In the $n\, =\, 6$ case, the difference between ${\cal AC}_{n-4}(T_{ij}/(i,j))$ and ${\cal PAC}_{n-4}$ which are located in $X_{ij}\, <<\, m^{2}$ is even more pronounced.  If we consider the associahedron block $A_{3}^{{\cal F}_{14}}$ and anticipate that the projection onto the $X_{14}\, <<\, m^{2}$ region is equivalent to considering the reference dissection $\{(1,3), (1,5)\, \}$ then we see an immediate contradiction. Starting with ${\cal D}\, =\, \{\, (1,3),\, (1,5)\, \}$ We obtain a two dimensional ${\cal AC}_{2}({\cal D})$ which is a  pentagon with five facets labelled by
\begin{flalign}
\{ (1,3), (2,6), (1,5), (2,6), (4,5)\, \}\nonumber
\end{flalign} 
However as there is no diagonal which does not intersect all of these five facets, this particular accordiohedron is not a (co-dimension one) face of the three dimensional associahedron. On the other hand, ${\cal PAC}_{2}[(1,4)]$ sitting in the $X_{14}\, <<\, m^{2}$ region, is by construction (parallel) to the $X_{14}\, =\, m^{2}$ facet of the associahedron block $A_{3}^{{\cal F}_{14}}$. 
\end{itemize}
In summary, we see that,\\
{\bf (i)} For all the accordiohedra which are in fact projections on to co-dimension one facets of an associahedra, we need to identify the reference dissection so that they are projected accordiohedra emerging in the low energy limit.\\ 
{\bf (ii)} And in fact, there is a class of accordiohedra that are defined using specific reference dissections consisting of $(n-4)$ triangles and one quadrilateral which can not be realised as projected accordiohedra. 

The second example cited above pertains to a larger structure of an accordiohedron. Namely, not all accordiohedra are facets of higher dimensional associahedra. It is hence puzzling to see that both the descriptions, one via accordiohedron as positive geometry and one via projected facets ${\cal PAC}_{n-4}$ of the associahedron blocks in the low energy limit can produce the same tree-level amplitude.

Resolution to this contradiction was already hinted at in \cite{mrunmay3} where it was noticed in certain examples that any combinatorial accordiohedron ${\cal AC}_{n-4}({\cal D})$ which can not be realised as a co-dimension one facet of a higher dimensional associahedron does not contribute to the  scattering amplitude! We prove this statement in  appendix \ref{OW}.

We will now show that, given an ${\cal F}_{ij}$ with red vertices co-ordinatized by ${\cal S}_{ij}\, =\, \{\, (i,j),\, \dots$ $\dots\, (i + \vert j - i\vert, j\, +\, \vert j - i\vert)\, \}$, any one of the projected accordiohedron  ${\cal PAC}_{n-4}[(p.q)]$ that emerge in $\{\, X_{p,q}\, <<\, m^{2}\, \vert\, (p.q)\, \in\, {\cal S}_{ij}\, \}$ region  can indeed be realised as ${\cal AC}_{n-4}({\cal D}_{p,q})$. Where ${\cal D}_{p ,q}$ is defined as follows.\\
\begin{flalign}\label{dmneqn}
D_{p.q}\, =\, \{\, X_{p,p-2},\, \dots,\, X_{p,q+1},\, X_{p+1, q},\, X_{p+2,q},\, \dots,\, X_{p+2,q+3},\, \dots,\, X_{p,p-2}\, \}
\end{flalign}
Before proving this claim we first verify it in a couple of examples.
\begin{itemize}
\item In the case of $n\, =\, 5$ and the fundamental domain ${\cal F}_{13}$ if we consider a reference dissection as $\{\, (1,3)_{\textrm{R}}, (1,4)\, \}$ then the resulting one dimensional projected accordiohedron in $X_{13}\, <<\, m^{2}$ domain is precisely a polytope with vertices labelled by $\{\, (1,4),\, (3,5)\, \}$.
\item In the same way in the case of $n\, =\, 6$ and fundamental domain ${\cal F}_{14}$ choosing the reference as $\{\, (2,4), (1,4), (1,5)\, \}$ and deleting the red diagonal $(1,4)$ in $X_{14}\, <<\, m^{2}$ limit results in a two dimensional accordiohedron with four vertices.
\end{itemize}
\begin{lemma}
Given a fundamental domain ${\cal F}_{ij}$ with $(p,q)$ co-ordinatizing one of the red vertices,  consider the projected accordiohedron ${\cal PAC}_{n-4}[(p,q)]$ obtained by analysing the associahdron block $A_{n-3}^{{\cal F}_{ij}}$ from $X_{pq}\, <<\, m^{2}$ region. We claim that 
\begin{flalign}
{\cal PAC}_{n-4}[(p,q)]\, =\, {\cal AC}_{n-4}(D_{p,q})
\end{flalign}
where ${\cal D}_{p,q}$ is defined in equation \eqref{dmneqn}.
\end{lemma}
{\bf Proof} : The proof is simply by inspection. We first analyze $(p,q)\, =\, (i,j)$ case. Consider the initial quiver in ${\cal F}_{ij}$ which is the reference triangulation $T_{ij}$ (generating the ABHY equations).
\begin{flalign}
T_{ij}\, =\, \{X_{i,i-1},\, \dots,\, X_{i,i+1}\, \}
\end{flalign}
Where we span the vertex label $i-1$ to $i+1$ counter-clockwise from $i$ and traverse through $X_{ij}$. ${\cal D}_{ij}$ is simply obtained from $T_{ij}$ by,
\begin{flalign}
D_{ij}\, =\, \cup_{m=j}^{i+1}\, T_{ij}\backslash(i,m)\, \cup\, (i+1, m)
\end{flalign}
A moment of meditation reveals that the all the diagonals which are not transversely intersect  $(i,j)$ diagonal are necessarily not compatible with the dissection ${\cal D}_{ij}$, (see figure \ref{lemmaproof}). This completes the proof.
\begin{figure}
    \centering
    \includegraphics[scale=0.35]{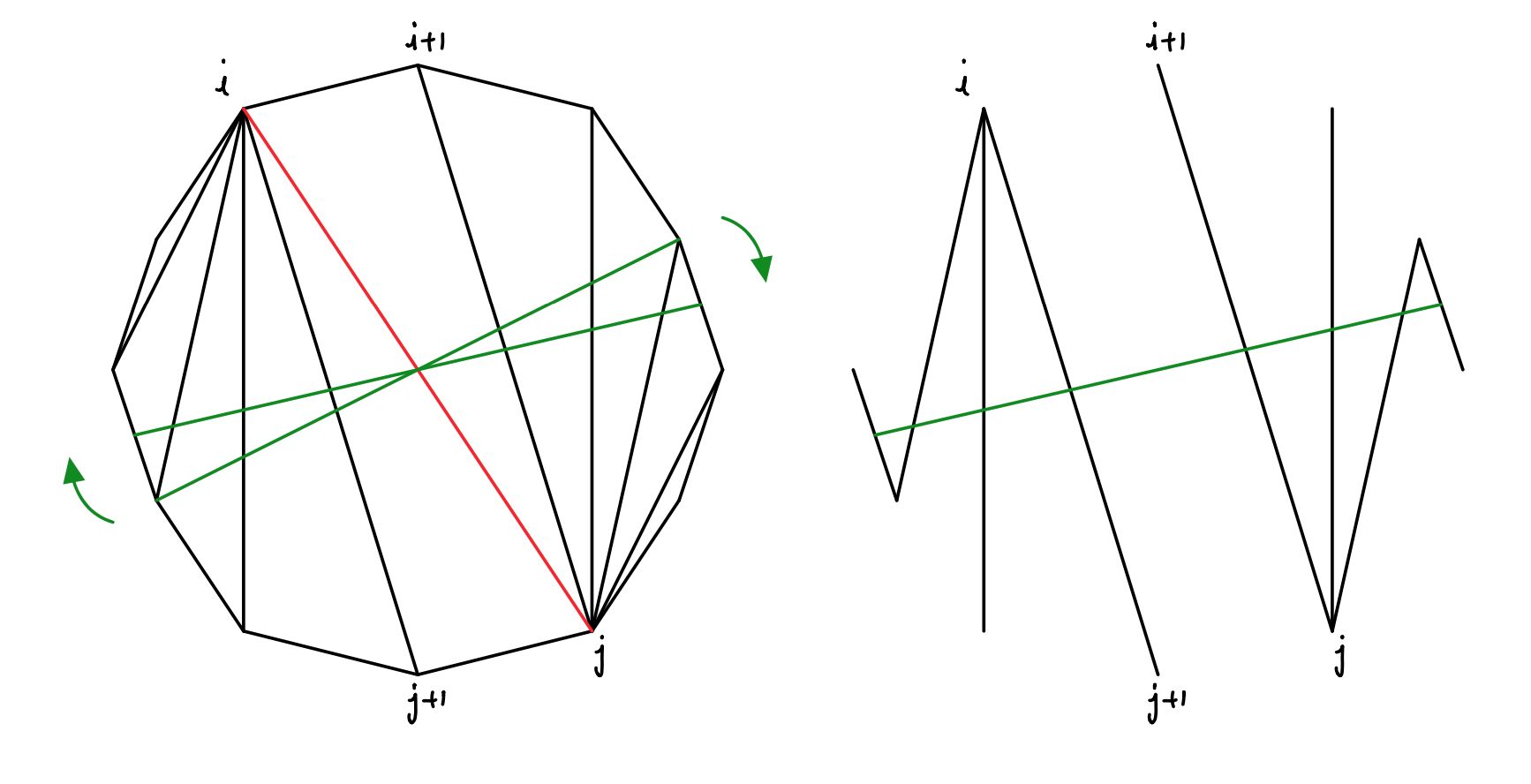}
    \caption{A diagonal intersecting $(ij)$ is not compatible with the dissection ${\cal D}_{ij}$}
    \label{lemmaproof}
\end{figure}

Similarly  ${\cal PAC}_{n-4}(p, q)$ (where $(p,q)\, =\, (i+k,j+k)$) obtained by projecting onto the face parallel to $\tilde{X}_{i+k,j+k}\, =\, 0$ is
isomorphic to the accoriohedron ${\cal AC}_{n-4}(D_{i+k,j+k})$ where $D_{i+k,j+k}$ is obtained by rigidly translating the quiver associated to $D_{ij}$ or equivalently rotating the dissection $D_{ij}$. 
Finally, the complete colour-ordered tree-level amplitude can be computed by using weighted sum over accordiohedra as in \cite{mrunmay3}. This result is simply based on the claim proved in the previous section. Thanks to the identification established through equation \eqref{dmneqn}, the planar scattering form $\Omega_{n-4}$ determined by ${\cal PAC}_{n-4}[(p,q)]$ equals the planar scattering form defined by ${\cal AC}_{n-4}({\cal D}_{p,q})$. Hence we have the following. 
\begin{flalign}\label{eftcf}
\Omega_{\textrm{EFT}}^{g^{2}}\, =\, \sum_{{\cal D}_{p,q}}\, \alpha_{[{\cal D}_{p,q}]}\, \Omega^{{\cal D}_{p,q}}_{n-4}\vert_{{\cal AC}_{n-4}({\cal D}_{p,q})}
\end{flalign}
where  $[{\cal D}]$ was introduced in section \ref{cfa}. It  is the equivalence class of (reference) dissection defined as the orbit of ${\cal D}$ under the dihedral group $D_{n}$.\\  
$\alpha_{[{\cal D}_{p,q}]}$ are determined in Appendix \ref{OW}.
\section{Discussion and Open Issues}\label{doi}
{\bf Perturbative expansion at higher orders}\\
A rather natural question that arises out of our analysis is the folliwng :  Can similar ideas be used to compute amplitude for $\lambda_{1}\, \phi_{1}^{3}\, +\, \lambda_{2}\, \phi_{1}^{2}\phi_{2}$ theory to arbitrary orders in $\lambda_{2}$, that is including amplitude contributions that have (at least) two massive and remaining massless poles. We can see that there is an immediate limitation to generalising the structure of coloured causal diamonds to define polytopes which includes such  contributions. Given any two diagonals $(ij),\, (ik)$ in an $n$-gon, there is always a complete triangulation that\\
{\bf (a)} includes $(ij),\, (ik)$  and,\\ 
{\bf (b)} the triangulation is such that these two diagonals are adjacent to each other.

Thus, irrespective of what mutation rules we define, (so as to define a fundamental domain which has an ``initial dissection" that is associated to a diagram with two massive propagators in our theory) this domain will always include a triangulation which is dual to the a Feynman diagram in  scalar theory with $\phi_{1}\, \phi_{2}^{2}$ vertex. Hence  positive geometries for higher order perturbative expansion of scattering amplitude in $\lambda_{1}\, \phi_{1}^{3}\, +\, \lambda_{2}\, \phi_{1}^{2}\, \phi_{2}$ theory remain unknown.  Based on simple examples (such as $n\, =\, 6$ case) we do expect the positive geometries to be open polytopes such as those discovered in \cite{songmatter}. However a detailed analysis of such ``open associahedron blocks" is outside the scope of this work.

\hspace*{-0.3in}{\bf Amplitude for most general cubic interactions}\\
Yet another related question which we do not answer in this paper is the following: Can the perturbative expansion of (planar and tree-level) scattering amplitude generated by the most general cubic interaction involving a massless and a massive scalar field be obtained from positive geometries located in positive regime of ${\cal K}_{n}$. We suspect that the answer is yes: By expanding the definition of fundamental domain to include ``un-coloured" triangulations more than once, we do end up including vertices of the associahedra that are dual to poles with the generic interactions. 

\hspace*{-0.3in}{\bf Relationship with CHY Formulation} :\\
One of the most striking consequences of the discovery of the kinematic space associahedron in \cite{abhy1711} was a derivation of the CHY formula for bi-adjoint $\phi^{3}$ theory. In \cite{abhy1711} it was shown that the compactification of the (real section) of the CHY moduli space, namely ${\bf M}_{0,n}({\bf R})$ was diffeomorphic to the kinematic space associahedron where the diffeomorphisms were defined by the scattering equations. It is thus a natural question to ask if the convex realisation of  associahedron blocks  are  diffeomorphic to the CHY moduli space with diffeomorphism defined by the scattering equations corresponding to mixed scalar interactions.  A proposal for this class of scattering equations is given in \cite{lam2005}.

In \cite{lam2005} a proposal was given for generalisation of CHY scattering equations to theories with several species of scalar fields. The numerator in the CHY scattering equations which depend on Mandelstam invariants are deformed based on whether the corresponding channel has a massive pole or a massless pole. Lam has given the CHY formula for $\phi_{1}^{2}\, \phi_{2}$ interaction  A detailed analysis of the relationship between CHY integrand for bi-scalar theory and weighted sum over d log forms is outside the scope of the paper and will be attempted elsewhere but we are tempted to offer some speculation.

If the (generalised) scattering equations proposed in \cite{lam2005} do turn out to be the diffeomorphisms which map the (compactified) CHY moduli space to the associahedron block then (1) this will imply that the weighted sum of  canonical forms on the   associahedron blocks is a pushforward of a CHY integrand and (2) the inverse of diffeomorphism between CHY moduli space and ABHY associahedron can be composed with the diffeomorphism from CHY moduli space to associahedron block to have a diffeomorphism between massless $\phi^{3}$ amplitude and perturbative amplitude (up to $\lambda_{2}^{\, 2}$) in two-scalar field theory. At this stage this is a pure speculation but if true it may give rise to interesting inter-relationships between tree-level amplitudes of physically inequivalent theories.

\hspace*{-0.3in}{\bf Perturbative Amplitude at All Orders}:\\
In this paper, we focused solely on tree level colour-ordered amplitude at order $\lambda_{2}^{\, 2}$. But our eventual goal is to find the positive geometry in kinematic space whose canonical form generate amplitude at any order in the coupling.  However there is a clear structural difference between  amplitude at order $\lambda_{2}^{\, 2}$ from higher order contributions. As we argued earlier, the polytope with two or more adjacent red facets are bound to be open if we fix the interaction between the two fields to be $\phi_{1}^{2}\phi_{2}$. Whether the causal structure in kinematic space encodes combinatorics and realisations of such polytopes remains to be seen.

This question is intimately tied to seeing how accordiohedra for perturbative amplitudes in EFT at higher order in $g$ emerge in the decoupling limit from positive geometry of the bi-scalar field theory. 

\hspace*{-0.3in}{\bf Coloured Causal diamonds and Positive geometries for several fields}:\\
Finally it would be interesting to see if the ideas proposed in this paper can be generalised suitably to find positive geometry for tree-level amplitudes of several fields.\footnote{We are grateful to Nima Arkani-Hamed for posing this question.} As ``friends" of associahedra with various possible colorings of the facets remain to be found, we hope that a colored causal structure in the kinematic space can be used to define such polytopes.

\hspace*{-0.3in}{\bf Local QFTS and Positive Geometries}\\
Our analysis centered around a specific mixed cubic interaction between a massless and a massive field and we computed tree-level amplitudes with only massless external particles.  It is an important question to analyse the correspondence between positive geometries in kinematic space and local unitary  scalar scattering amplitudes. That is, how do we classify all the local multi-scalar interactions for which \emph{any} scattering amplitude equals sum over $d\log$ forms associated to positive geometries. We should note that computation of generic bi-adjoint amplitude with a single massless scalar or the amplitudes in the Yukawa type interaction analysed in \cite{songmatter} require us to consider so-called open polytopes in which some of the facets lie at infinity in the kinematic space. Hence our organisation principle should allow for open associahedra (or more generally open polytopes) as possible positive geometries for scattering amplitudes.

We hope to return to at least some of these questions in the future.
\appendix
\section{Explicit formula for the weights}\label{OW}
In this appendix, we derive an explicit formula for the weights $\alpha_{[{\cal D}_{p,q}]}$ that appeared in equation \eqref{eftcf}. We then prove that  if a combinatorial accordiohedron ${\cal AC}({\cal D})$ of dimension $n-4$ is not (combinatorially) equivalent to one of the co-dimension one facet of $A_{n-3}$ then the weight $\alpha_{[{\cal D}]}$ is zero.

The explicit formula for the weights when the accordiohedron polytope is generated a reference dissection consists of $n\, -\, 4$  triangles and one quadrilateral is derived as follows. 
\begin{lemma}
Consider an accordiohedron whose reference dissection consists of $n - 4$ triangles and one quadrilateral.  If the accordiohedron is a facet of associahedron associated to one partial triangulation obtained by removing $(i,j)$ then,
\begin{flalign}
\begin{array}{lll}
\alpha_{[{\cal D}_{i,j}]}\, =\, \frac{1}{2 C_{k-2}\, C_{n-k-2}}\, \textrm{if}\ \vert\, j - i\vert\, =\, k\, \textrm{modulo}\, n
\end{array}
\end{flalign}
\end{lemma}
where $C_{l}$ is the Catalan Number defined as
\begin{flalign}
C_{l}\, =\, \frac{(2l)\, !}{l\, !\, (l+1)\, !}
\end{flalign}
Note that if $\vert i - j\vert\, =\, 2\, \textrm{modulo}\, n$ then the above formula reduces to,
\begin{flalign}
\alpha_{[{\cal D}_{i,j}]}\, =\, \frac{1}{2\, C_{n-4}}\, \end{flalign}
\begin{proof}
Let $(\partial A)_{n-4}(i,j)$ be one of the facets  of $A_{n-3}$ that corresponds to one-partial triangulation obtained by removing the diagonal $(i,j)$. 
$(\partial A)_{n-4}(i,j)$ corresponds to a set of $\{\, {\cal AC}(B^{\prime}_{ij})\, \}$ where $B^{\prime}_{ij}$ is the set of all dissections which lead to the same combinatorial accordiohedron.\\
Now let $\vert j - i\vert\, =\, 2\, \textrm{modulo}\, n$. In this case the reader can convince themselves easily that all the triangulations of the $n-1$ gon obtained by chopping off the triangle $i,\, i + 1\, j$ lead to the same accordiohedron. Hence the number of such accordiohedra is $C_{n-4}$. Similarly if $(\partial A)_{n-4}(i,j)\, =\, A_{r}\, \times\, A_{n-r-4}$ then the total number of dissections which lead to the this facet as an accordiohedron is $C_{r}\, \times\, C_{n-r-4}$.\\
We also note that each reference dissection $B^{\prime}_{ij}$ leads to an  accordiohedron which is one of two possible facets :  $(\partial A)_{n-4}(i,j)$ and either $(\partial A)_{n-4}(i + 1, i\, + 1)$ or  $(\partial A)_{n-4}(i - 1, i\, - 1)$ and hence weight that we assign to each boundary accordiohedron for a given $(i,j)$ is $\frac{1}{2\, C_{r}\, C_{n-r-4}}$. This completes the proof.
\end{proof}


\section*{Acknowledgement}
We are indebted to Nima Arkani-Hamed for a number of insightful discussions, many clarifications and encouragement.\\ 
We would like to thank  Ashoke Sen and Nemani Suryanarayana for valuable inputs and Sujay Ashok, Pinaki Banerjee, Miguel Campiglia, Dileep Jatkar, Nikhil Kalyanapuram, Madhusudan Raman, Prashanth Raman and Arnab Priya Saha for many  discussions over the years on related issues. We also thank Pinaki Banerjee for comments on the manuscript.\\
We would especially like to thank Vincent Pilaud for his guidance and crucial insights in the early stages of this work.


\end{document}